\documentclass[acmsmall,screen,nonacm]{acmart}

\usepackage{blkarray}
\usepackage{amsthm}
\usepackage{amsmath}
\usepackage{amssymb}
\usepackage{bbm}
\usepackage{color}
\usepackage{xcolor}
\usepackage{hyperref}
\hypersetup{
    colorlinks=true, 
    linktoc=all,     
    linkcolor=blue,
}
\usepackage{tikz}
\usetikzlibrary{shapes,arrows.meta}
\usepackage{cleveref}
\usepackage{mathpartir}
\usepackage{multirow}
\usepackage{subcaption}
\usepackage{xfrac}
\usepackage{listings}
\definecolor{eclipseBlue}{RGB}{42,0.0,255}
\definecolor{eclipseGreen}{RGB}{63,127,95}

\DeclareRobustCommand{\stirling}{\genfrac\{\}{0pt}{}}

\newcommand{\econd}[3]{\ensuremath{\mathtt{if}\;#1\;\mathtt{then}\;#2\;\mathtt{else}\;#3}}

\newcommand{\etrue}{\ensuremath{\mathtt{true}}}
\newcommand{\efalse}{\ensuremath{\mathtt{false}}}

\newcommand{\elet}[3]{\ensuremath{\mathtt{let}\;#1\;=\;#2\;\mathtt{in}\;#3}}

\newcommand{\efun}[3]{\ensuremath{\mathtt{fun}\;#1\;#2\;=\;#3}}

\newcommand{\eapp}[2]{\ensuremath{#1\;#2}}

\newcommand{\enil}{\ensuremath{\mathtt{[]}}}

\newcommand{\econs}[2]{\ensuremath{#1\,\mathtt{::}\,#2}}

\newcommand{\ecasel}[5]{\ensuremath{\mathtt{case}\;#1\;\mathtt{of}\;
\enil \rightarrow #2
\mid \econs{#3}{#4} \rightarrow #5}}

\newcommand{\evar}[1]{\ensuremath{\mathtt{#1}}}

\newcommand{\vclosure}[4]{\ensuremath{\mathtt{C}(#1;\;#2,\,#3.\,#4)}}

\newcommand{\tbool}{\ensuremath{\mathbb{B}}}

\newcommand{\tlistshow}[2]{\ensuremath{L^{#1}(#2)}}
\newcommand{\tlist}[1]{\ensuremath{L(#1)}}

\newcommand{\tprod}[2]{\ensuremath{#1 \otimes #2}}

\newcommand{\tfunshow}[2]{\ensuremath{#1 \rightarrow #2}}
\newcommand{\tfun}[3]{\ensuremath{#1 \stackrel{#3}{\rightarrow} #2}}
\newcommand{\potp}[2]{\ensuremath{\langle #1;\,#2\rangle}}

\newcommand{\inds}[1]{\ensuremath{\mathit{Ind}(#1)}}

\newcommand{\ideg}[1]{\ensuremath{\mathtt{d_{#1}}}}
\newcommand{\ibase}[1]{\ensuremath{\mathtt{b_{#1}}}}
\newcommand{\iarg}{\ensuremath{\mathtt{a}}}
\newcommand{\iret}{\ensuremath{\mathtt{r}}}
\newcommand{\iconst}{\ensuremath{\mathtt{c}}}

\newcommand{\idot}{\ensuremath{.}}

\newcommand{\wf}[2]{\ensuremath{#1 : #2}}

\newcommand{\pot}[3]{\Phi (#1 : #2 \mid #3 )}

\newcommand{\dom}[1]{\ensuremath{\texttt{dom}(#1)}}

\newcommand{\ejudgenp}[3]{#1 \vdash #2 \Downarrow #3}
\newcommand{\mi}{I}
\newcommand{\mshift}[2]{\lhd^{#1}_{#2}}
\newcommand{\munshift}[2]{\rhd^{#1}_{#2}}

\newcommand{\mmove}[2]{\mathit{Move}^{#1}_{#2}}

\newcommand{\mnil}[1]{\mathit{Nil}_{#1}}

\newcommand{\mproj}[1]{\pi_{#1}}
\newcommand{\mzero}[1]{\mproj{\neg #1}}

\newcommand{\tjudgecf}[5]{#1 \vdash_{\mathit{cf}} #2 : #3 \rightsquigarrow #4 \mid #5}

\newcommand{\subsubsubsection}[1]{\smallskip\noindent{\bf #1.}}

\AtBeginDocument{%
  }

\setcopyright{acmlicensed}
\copyrightyear{2018}
\acmYear{2018}
\begin{document}

\title{Efficient Cost Bounds with Linear Maps}

\author{David M Kahn}
\email{davidkah@cs.cmu.edu}
\affiliation{%
  \institution{Carnegie Mellon University}
  \streetaddress{5000 Forbes Avenue}
  \city{Pittsburgh}
  \state{Pennsylvania}
  \country{USA}
  \postcode{15213}
}

\author{Jan Hoffmann}
\affiliation{%
  \institution{Carnegie Mellon University}
  \streetaddress{5000 Forbes Avenue}
  \city{Pittsburgh}
  \state{Pennsylvania}
  \country{USA}
  \postcode{15213}
}
\email{janh@cs.cmu.edu}

\author{Thomas Reps}
\affiliation{%
  \institution{University of Wisconsin}
  \city{Madison}
  \state{Wisconsin}
  \country{USA}
}
\email{reps@cs.wisc.edu}

\author{Jessie Grosen}
\affiliation{%
  \institution{Carnegie Mellon University}
  \streetaddress{5000 Forbes Avenue}
  \city{Pittsburgh}
  \state{Pennsylvania}
  \country{USA}
  \postcode{15213}
}
\email{jgrosen@cs.cmu.edu}

\renewcommand{\shortauthors}{Kahn et al.}

\begin{abstract}
  The Automatic Amortized Resource Analysis (AARA) 
  derives program-execution cost bounds using types.
  To do so, AARA often makes use of 
  \textit{cost-free} types, which are critical for the composition 
  of types and cost bounds.
  However, inferring cost-free types using the
  current
  state-of-the-art algorithm is expensive due to recursive dependence on additional cost-free types. 
  Furthermore, that 
  algorithm uses a heuristic only
  applicable to polynomial cost bounds, and not, e.g., 
  exponential bounds.
  
  This paper
  presents a new approach to these problems 
  by representing the 
  cost-free types of a function in 
  a new way: with a
  linear map, which can stand for infinitely many cost-free types.
  Such maps enable an algebraic flavor of reasoning about cost bounds
  (including non-polynomial bounds)
  via
  matrix inequalities. 
  These inequalities can be solved with off-the-shelf linear-programming tools for many programs, so that types can always
  be efficiently checked and often be efficiently inferred.
  An experimental evaluation with a prototype
  implementation shows
  that---when it is applicable---the
  inference of linear maps is exponentially more efficient than the state-of-the-art algorithm.
\end{abstract}



\received{20 February 2007}
\received[revised]{12 March 2009}
\received[accepted]{5 June 2009}

\maketitle

\section{Introduction}

There is an inherent tradeoff between the expressivity and 
level of automation present in the cost-analysis tools of today. 
Tools that drop automation in favor of human guidance
\cite{amadio2014certified,gueneau2018fistful,niu2022cost}
can support a wide variety of cost bounds
for just about any sort of code pattern. On the flip side,
those tools that can operate
(semi-)automatically \cite{albert2008automatic,albert2012cost,gulwani2009speed,klemen2023solving,rebon2002cost} must make concessions between the
kinds of bounds and code patterns they can support and the
efficiency with which they can do so.

The subject of this paper, Automatic Amortized Resource 
Analysis (AARA) \cite{hofmann2003static}, belongs to the automatic camp.
AARA is a type system for functional programs
that uses its types 
to mediate an application of 
the physicist's method of amortized cost 
analysis \cite{kn:Tarjan85}.
This approach posits
that data structures---such as lists and trees---carry 
potential, and this potential is used to pay for costs.
The types of functions $\tau \rightarrow \sigma$ 
then describe their execution 
costs through the change in potential between 
$\tau$ and $\sigma$. The amount of potential 
on a data structure is given by a type annotation, and
typing a program with AARA induces 
linear constraints between these annotations, 
which reduces the problem 
of finding a cost bound to linear programming.
While the initial incarnation of this approach only found
linear bounds on heap space
\cite{hofmann2003static}, 
later extensions use the same basic approach to automatically 
infer polynomial \cite{hoffmannpolynomial,HoffmannAH12},
exponential \cite{kahn2020exponential},
or logarithmic bounds \cite{hofmann2022type} 
for
user-defined cost models, such as time, space, or energy.
Further work has also adapted AARA to 
work on probabilistic \cite{wang2020raising},
higher-order \cite{Jost10},
and parallel programs \cite{hoffmann2015automatic} and to 
domains including smart contracts \cite{das2018work} and 
CUDA kernels \cite{muller2021modeling}.

All of AARA's automation comes with tradeoffs.
Our work aims to improve one of the tradeoffs AARA makes 
between expressivity and efficiency.
To support
the inference of nonlinear cost bounds 
for non-tail recursive functions,
previous work on AARA
\cite{hoffmann2010amortized} 
introduced the special \textit{cost-free} type.
Cost-free types are just like normal (costful) AARA types,
except that their potential is \textit{never} used to pay for costs;
instead the potential just moves around between data structures.
The role of cost-free types is to express how \emph{excess} potential is
% This leaves cost-free types perfect for expressing how excess potential gets
allocated, which is critical for supporting nonlinear cost bounds 
for non-tail-recursive functions, as well as 
generally enabling useful compositional cost bounds.
However, in exchange for this extra expressivity, working with AARA types becomes more difficult.
In particular, the state-of-the-art algorithm for inferring
cost-free types \cite{hoffmann2010amortized,HoffmannAH12} has
two main problems: it scales exponentially poorly,
and it does not extend to non-polynomial cost bounds.
This paper aims to alleviate both problems via a new
method of handling cost-free types.

To clarify where our work impacts the AARA analysis,
we now explain (one way) 
how cost-free types fit into AARA. 
Simplified AARA type inference follows the solid arrows
in \Cref{fig:diagram}: take the code of a function;
infer the function's \textit{costful} typing;
and then interpret this type with physicist's method of reasoning to obtain cost bounds on running the function. 
While this simplified picture is basically correct,
it omits a major component contributing to the quality of the 
resulting cost bounds: \textit{cost-free types} \cite{hoffmann2010amortized}. 
Cost-free types describe how extra potential should be conserved,
which can be applied to keep cost bounds tight.
A good cost\textit{ful} type for a function can depend on a
cost-\textit{free} type for that same function. Further,
a good cost-free type can depend recursively on other
cost-free types in an infinite regress known as 
\textit{resource-polymorphic recursion} (\Cref{sec:rpr}). 
Our work aims at replacing
this expensive recursive dependence with 
a direct and efficient solution.

\begin{figure}
\scalebox{0.9}{
\begin{tikzpicture}[>=Stealth, node distance=2.5cm, font=\sffamily]
    % Nodes
    \node (source) [draw, rounded rectangle] {function code};
    \node (costfree) [draw, rounded rectangle, above=1cm, right=0.5cm] {cost-free typing};
    \node (costful) [draw, rounded rectangle, right=2cm] {costful typing};
    \node (costbounds) [draw, rounded rectangle, right=5cm] {cost bounds};

    % Arrows
    \draw[->,line width = 1pt] (source) -- (costful);
    \draw[->,line width = 1pt] (costful) -- (costbounds);
    \draw[dotted,->,line width = 1pt] (source) -- (costfree);
    \draw[dotted,->,line width = 1pt] (costfree) -- (costful);
    \draw[dotted,->,line width = 1pt] (costfree) to [out=45,in=135,looseness=4](costfree);
\end{tikzpicture}
}
\caption{Diagram of how cost-free typing fits into AARA}
\label{fig:diagram}
\end{figure}
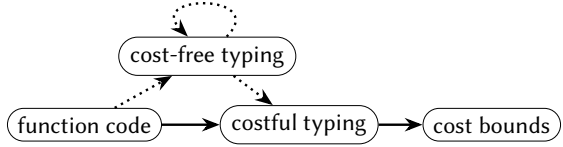

The current cost-free 
inference algorithm \cite{hoffmann2010amortized,HoffmannAH12} is essentially brute-force iteration: 
a function gets retyped at each application
site to specialize the function's type to its argument's 
potential.\footnote{
Or symmetrically, specialize to the function's return.
}
This approach only avoids infinite regress by
introducing an extra assumption
about the problem structure that guarantees type checking terminates in a finite number of steps.
However, this assumption is based on the
discrete differences of polynomials,
and thus fundamentally is only successful with polynomial potential.
Even then, the assumption does not always hold 
for polynomial potential (see \Cref{sec:old}).
Furthermore, this iterative approach can be quite slow.
Inferring polynomial cost bounds of degree $d$ for a function $f$ that recursively calls itself at $b$ places in its body results in retyping $f$ almost $b^d$ times.
This process is then repeated to obtain a new typing for $f$ every time $f$ is called outside its body.
Moreover, this process combinatorially explodes if $f$ calls 
some other function $g$, because $g$ might itself need to be 
retyped exponentially many times each time $f$ is retyped.
We analyze these blowups more formally in \Cref{sec:exp}. 
Even though linear constraint solving 
can run in polynomial time, running 
this iterative cost-free
typing algorithm can make AARA exponentially slow
in the size of the source code because
the number of generated constraints can grow
exponentially.

In contrast, our new approach handles cost-free types using a
\emph{functional transformation}---specifically a linear map---from 
the potential on the argument type $\tau$ to that of the return type $\sigma$. Such a function can be found by solving matrix 
inequations, which can be efficiently reduced to linear programming
for many common code patterns, such as functions with linear recursion
schemes.
This approach sidesteps resource-polymorphic 
recursion by uniformly covering every level of the infinite regress
with the same functional transformation. Unlike
the existing approach,
our new
approach also works on non-polynomial cost bounds, such as exponential bounds.

Each approach has its pros and cons:
our new approach is far more efficient, and is able to 
handle some nonpolynomial cost bounds;
however, the older approach can handle more common program patterns;
and each can infer cost-free types that the other cannot.
We discuss these issues in more detail
in \Cref{sec:exp,sec:further,sec:lim}.
We therefore suggest that both approaches be used in conjunction with each other to efficiently obtain cost-free types where possible, while
still maintaining high program coverage. Not only can our new functional
approach handle many programs directly, but it also can cut short
the expensive cascade of retypings performed
by the older iterative approach.

The contributions of our work are as follows: 
\begin{itemize}
  \item
    We present a declarative type system for cost-free typing using linear-map-based types, along with a proof of soundness.
    This system is the first that
    can give finite derivations of non-polynomial cost-free types in the presence of resource-polymorphic recursion.

  \item We show not only that type \textit{checking} in our 
    system can be automated using only simple 
  matrix operations, but also that
  type \textit{inference} can be automated by
    linear programming under common conditions.
    A sufficient condition is, for example, that the analyzed function 
    is linearly recursive 
    (makes at most one recursive call
    each time it executes its body). 
  
  \item We implemented a prototype of the
  inference algorithm
  and compared the prototype's performance 
  to the existing cost-free algorithm.
  We found that our new linear-map approach scales 
  exponentially better.
  For some programs our prototype takes a fraction
  of a second, while the old approach takes multiple days.

\end{itemize}

Of further note is that the impact of this 
work is not confined
merely to resource-polymorphic recursion.
Cost-free types also indirectly 
provide information about \emph{changes} in the sizes
of data structures: roughly,
if potential density doubles 
but the total potential remains unchanged,
then the size of the container of the 
potential must be halved.
This indirect way of gleaning information 
is a key part of the type systems of various
extensions of AARA, including that for 
parallel programs \cite{hoffmann2015automatic} 
and logarithmic cost bounds \cite{hofmann2022type}.

\subsubsubsection{Organization}
The remainder of the paper is organized as follows:
\Cref{sec:overview} presents an overview of
the AARA type system, the technical problem 
of \textit{resource-polymorphic recursion}, and
solution approaches using \textit{cost-free} types. 
\Cref{sec:ty} presents all the pieces of 
our cost-free system. \Cref{sec:sound}
states our soundness theorem. \Cref{sec:aut} describes
our approach to automation.
\Cref{sec:exp,sec:further} 
provide experimental and theoretical comparisons, respectively, between our new technique and the existing solution.
\Cref{sec:lim} discusses the limitations of our work
and obstacles to future extensions.
\Cref{sec:relatedwork} describes related work.
\Cref{sec:conc} concludes.

\section{Overview}
\label{sec:overview}

In this section, we give an overview of Automatic Amortized Resource Analysis, cost-free typing, and 
resource-polymorphic recursion. We explain both
the existing algorithm for cost-free typing, 
and give the intuition of our new algorithm.

\subsection{AARA}
\label{sec:aara}

The Automatic Amortized Resource Analysis (AARA) type system 
performs the physicist's method of amortized analysis by 
using types to assign potential. 
By letting the potential assigned to a
list be some polynomial function
of the list's length, AARA can infer polynomial cost bounds. 
The power of AARA comes 
from reducing the question of 
choosing a suitable potential function
to a linear programming problem.
A similar idea works for exponential \cite{kahn2020exponential}
and logarithmic bounds \cite{hofmann2022type}.

One simple kind of potential function is the linear function,
which represents some fixed amount of potential
per list element. Such a potential function can be 
encoded by augmenting each
list type with a number, like ``$5$'' in $\tlistshow 
{5} \tau$. This annotation indicates that the list 
holds potential equal to 5 times its length. We may 
also pair such types with some amount of constant potential 
like $\potp {\tlistshow {5} \tau} 3$, which
indicates
that there are
3 additional units of free potential that can pay for costs.

\begin{figure}
\begin{lstlisting}
    fun id lst = case lst of    (* free : 0; lst : 1 per elt *)
        | [] -> []              (* free : 0; [] : 0 per elt *)
        | x::xs ->              (* free : 1; xs : 1 per elt *)
            let () = tick 1 in  (* free : 0; xs : 1 per elt *)
            let xs' = id xs in  (* free : 0; xs' : 0 per elt *)
            x::xs'              (* free : 0; ret : 0 per elt *)
\end{lstlisting}
\caption{Code and accounting for id}
\label{fig:idcode}
\end{figure}

To see (a simple case of) AARA in action, consider 
the code in \Cref{fig:idcode} for 
\texttt{id}:$\tlist{\alpha} \rightarrow \tlist{\alpha}$, which is
the identity for lists. The number of recursive calls this code 
makes is equal to the length of the input list.
The ``tick'' expression
assigns a cost of 1 before the recursive call to
represent
this cost model.
Informally, we have used comments in \Cref{fig:idcode} to track how AARA assigns potential at each line of the code to cover this cost.
In particular,
just before line 4,
(i) we have 1 free unit of potential to pay for the recursive call, and (ii) \texttt{xs} is typed with the proper amount of potential to be an argument
of the recursive call.

More formally, AARA types $\texttt{id}$ as
$\potp {\tlistshow 1 \alpha} {0} \rightarrow \potp {\tlistshow 0 \alpha} {0}$,
which is annotated, not to carry potential, but rather to 
describe how potential is transformed.
Here the argument type $\potp {\tlistshow 1 \alpha} {0}$ represents
a list carrying 1 unit of potential per element and 0 free potential, and the return type $\potp {\tlistshow 0 \alpha} {0}$ represents returning a list with no potential.
If we now take the difference in potential between the input 
and output, the physicist's method 
gives us a bound on the net cost.
In this case, that difference is equal to the length of the 
input list, which is the exact cost.

However, this type is not the only type AARA could assign to 
\texttt{id}. AARA could also type \texttt{id} with any type 
that takes at least as much potential in, and pays at least as much in total.
%\twr{I think ``in total'' would be better than ``in net,'' which has an awkward ring.}
Specifically, AARA could assign
any type of the form $\potp {\tlistshow a \alpha} {b} \rightarrow \potp {\tlistshow c \alpha} d$ where $a-1 \geq c \geq 0$ and $b \geq d \geq 0$. 
All of these types provide 
valid upper bounds on the net cost, but they differ in that they allow for 
different amounts of potential to be initially present. AARA infers
whichever type is best by using a linear program, allowing 
the type system to specialize its function types to 
the amount of potential present at call sites.
This approach yields tight composition of cost bounds---at the cost of 
retyping the function for each call site. This retyping is the inefficiency that our work tackles.

More expressive cost bounds, such as polynomials, can be represented by annotating list types with vectors\footnote{ 
    In other AARA literature, these annotation vectors are usually written in the reverse order from our choice of order here.
} rather than single numbers \cite{hoffmannpolynomial}. Such a vector 
provides (non-negative) coefficients for a linear combination of
basis functions making up the desired potential function. For example,
the type $\tlistshow {3,1,4} \alpha$
represents the following
polynomial function of the list length $n$: 
$3 \binom n 3 + \binom n 2 + 4 \binom n 1$,
where $\binom n k$ is the binomial coefficient,
which counts
the number of 
ways to select $k$ elements out of $n$.
In general, $\tlistshow {\vec a} \alpha$ would represent 
the potential function
$\sum_i \vec a_i \binom n {k-i}$ where $k=|\vec a|$.  
Note that we use binomial coefficients like $\binom n k$ (which
is a polynomial of degree $k$) rather 
than standard polynomial basis elements like $n^k$.
We will often refer to the binomial coefficient potential function 
$\binom n k$ as degree $k$ potential, so that, e.g., $\binom n 2$ 
is quadratic potential.
The binomial-coefficient representation
provides a number of conveniences, such as
allowing the use of Pascal's identity 
$\binom {n+1} {k+1} = \binom {n} {k+1} + \binom n k$ to 
relate the annotations on a list like \texttt{x::xs} of
length $n+1$ and its tail \texttt{xs} of length $n$. For example,
applying Pascal's identity can show that potential of  
$ \mathtt{x::xs} : \potp {\tlistshow {3,1,4} \alpha} 0$
is equal to that of 
$ \mathtt{xs} : \potp {\tlistshow {3,4,5} \alpha} 4$.
Furthermore, Pascal's identity is a linear recurrence, which allows 
polynomial AARA type
inference to reduce to linear programming.

A different interpretation of annotation vectors can be used to
represent exponential cost bounds in a very similar 
way \cite{kahn2020exponential}. 
Instead of using binomial coefficients as a basis, one can
use Stirling numbers of the second kind 
$\stirling {n+1} {k+1} \in \Theta ((k+1)^n)$, which
count the number of ways to partition a subset of $n$
elements into $k$ partitions. 
We will often refer to the Stirling-number potential function 
$\stirling {n+1} b$ as base $b$ potential.
Thus, in this setting,
$\tlistshow 3 \alpha$ represents a list with 3 units of base 
2 potential,\footnote{
Note that base 1 potential is skipped because
base 1 is the constant function
already represented by $\iconst$.}
i.e.,
$3\stirling {n+1} 2$ where $n$ is the length of the list.
The general case in this setting is that
$\tlistshow {\vec a} \alpha$ represents
the exponential potential function
$\sum_i \vec a_i \stirling {n+1} {k+1-i}$ where $k=|\vec a|$.
Instead of Pascal's identity,
Stirling numbers admit the similar linear recurrence identity $\stirling {n+1} {k+1} = (k+1)\stirling n {k+1} + \stirling n k$. 

Despite the similarity between the polynomial and exponential settings,
previous work on polynomial cost-free typing does not
extend to the exponential setting. 
We show why this difference occurs in \Cref {sec:old}.
In contrast, our new system handles both kinds of potential uniformly:
our system can support either polynomial or exponential cost bounds by swapping a few primitive matrices (\Cref{sec:ty});
our soundness theorem is agnostic to the choice (\Cref{sec:sound});
and our new inference system (\Cref{sec:aut}) is able to infer tight cost bounds for $\texttt{half}$ in either setting.

\subsection{Resource-Polymorphic Recursion}
\label{sec:rpr}

\emph{Resource-polymorphic recursion} refers to situations in which
the type annotations required of a recursively called function 
differ from those used at the function entry.\footnote{
  This could be characterized as a sort of frame problem
  like that solved by separation logic's frame rule \cite{o2001local}.
}
This situation typically occurs when the call is non-tail recursive 
and uses non-linear potential (but may still occur 
even with only linear potential \cite{hoffmann2010amortized}).

Resource-polymorphic recursion is tricky in that,
naively, typing it requires an infinite family of 
types. To see this phenomenon in action, consider
what happens to quadratic potential in
the function \texttt{half} for halving the length of a list
from $n$ to $\lfloor n/2 \rfloor$, given 
in \Cref{fig:half}.\footnote{
Note that \texttt{half} here is \textit{not} a linear-cost 
function because it is not instrumented with \textit{any}
``tick'' expressions, 
let alone those that correspond to the runtime cost model. The function
therefore has no cost and the question
is entirely about how potential should be moved around on data
structures, rather
than about how potential is spent to cover costs.
}

\begin{figure}
\centering
  \begin{lstlisting}[xleftmargin=4.0ex]
    fun half lst = case lst of 
      | [] -> []
      | x1::xs1 -> case xs1 of
        | [] -> []
        | x2::xs2 -> let tmp = half xs2 in x1::tmp
  \end{lstlisting}
  \caption{Code for \texttt{half}}
  \label{fig:half}
\end{figure}

We can start by reasoning semantically about 
how we expect the potential should be transformed by the 
function \texttt{half}. 
Initially, let us type the function \texttt{half} to take in 
a list $\mathtt{lst}:\tlistshow {1,0} \alpha$ with 1
unit of quadratic potential, 
i.e., $\binom {n} 2$
total potential. The code contains no ``tick'' expressions,
so no potential is ever consumed to pay for costs.
If no potential is lost, then a good AARA type should attempt to
reallocate all the potential from the input to
the output, and should thus
re-express the amount in terms of $\lfloor n/2 \rfloor$ 
because the output list is half the length of the input. 
This goal is accomplished best 
when the output list is given the 
type $\tlistshow {4,1} \alpha$, having $4$ units of quadratic 
potential and $1$ unit of linear potential, because $\binom {n} 2 \geq 4\binom {\lfloor n/2 \rfloor} 2 + \binom {\lfloor n/2 \rfloor} 1$, 
and this inequality is actually tight for even $n$.
So we would like to show that $\mathtt{half}:\tfunshow {\potp {\tlistshow {1,0} \alpha} 0}
{\potp {\tlistshow {4,1} \alpha} 0}$.

Now let us compare this result to how potential is actually
transformed by AARA's typing rules in lines 3 to 5 of \Cref{fig:half}:

\begin{itemize}
    \item From line 3 to the start of line 5, the input list \texttt{lst} gets 
    broken up into two head elements (\texttt{x1} and \texttt{x2})
    and a tail list \texttt{xs2}.
    To transfer all potential to the tail,
    AARA makes use of Pascal's identity twice:
    $\binom {m+2} {k+2} = \binom {m+1} {k+2} + \binom {m+1} {k+1}
    = \binom {m} {k+2} + 2\binom {m} {k+1} + \binom m k$.
    Reallocating potential according to this rule
    turns 1 unit of quadratic potential on \texttt{lst}
    into 1 unit of quadratic potential, 2 units of linear potential,
    and 1 unit of constant potential on \texttt{xs2}, because 
    \texttt{xs2} has
    two
    fewer elements than \texttt{lst}. 
    This
    allocation
    results in the type of \texttt{xs2} being 
    $\tlistshow {1,2} \alpha$ and 1 unit of free potential being leftover.
    This result makes sense because
    $\binom {m} 2 + 2\binom {m} 1 + 1 = \binom {m+2} 2$.

    \item The next part of line 5 performs a recursive call 
    on \texttt{xs2} and binds it to \texttt{tmp}. However, the function type assigned to 
    \texttt{half} by our above semantic reasoning has
    an argument type of $\tlistshow {1,0} \alpha$, 
    while $\mathtt{xs2}:\tlistshow {1,2} \alpha$. 
    AARA needs a different type for 
    \texttt{half} to justify the current one!
    It turns out that the required new typing for \texttt{half} 
    is $\tfunshow {\potp {\tlistshow {1,2} \alpha} 1}
    {\potp {\tlistshow {4,5} \alpha} 1}$.
     We can see that this is semantically valid by noting that
    $\binom {2m} 2 + 2\binom {2m} 1 +1 = 4\binom {m} 2 + 5\binom {m} 1 + 1$.

    \item Finally, the rest of line 5
    adds \texttt{x1} onto the front of 
    the list.
    In terms of potential, this action corresponds to 
    inverting Pascal's identity. This process
    leaves the output list typed $\tlistshow {4,1} \alpha$
    as desired, because 
    $4\binom {m} 2 + 5\binom {m} 1 + 1
    =4\binom {m+1} 2 + 1\binom {m+1} 1$.
    
\end{itemize}

This typing therefore works out consistently with our 
semantic reasoning, at least as long as we have that additional 
typing for \texttt{half} at the recursive call. However, it is 
justifying that additional typing $\mathtt{half}:\tfunshow {\potp {\tlistshow {1,2} \alpha} 1} {\potp {\tlistshow {4,5} \alpha} 1}$ that is the problem.
If we naively try to justify this 
new type syntactically with AARA's typing rules, 
we would find that we would need another 
new type {$\mathtt{half}:\tfunshow {\potp {\tlistshow {1,4} \alpha} 6}{\potp {\tlistshow {4,9} \alpha} 6} $ at the point of the recursive call. Justifying this type 
would then need yet another new type $\mathtt{half}:\tfunshow {\potp {\tlistshow {1,6} \alpha} {15}}  {\potp {\tlistshow {4,13} \alpha} {15}}$,
which would need still another type $\mathtt{half}:\tfunshow  {\potp {\tlistshow {1,8} \alpha} {28}} {\potp {\tlistshow {4,17} \alpha} {28}}$,}
and so on leading to infinite regress. 

Resource-polymorphic
recursion arises the same way when using exponential potential instead of
polynomial. In the exponential setting, one can
use the identity 
$\stirling {2n+1} 2 = 6\stirling {n+1} 4 + 6 \stirling {n+1} 3 + 3 \stirling {n+1} 2$ to guess that \texttt{half} can be
typed as $\tfunshow {\potp {\tlistshow {0,0,1} \alpha} {0}}  {\potp {\tlistshow {6,6,3} \alpha} {0}}$. However,
just like in the polynomial case, a different type for 
\texttt{half} is needed to
justify the typing of the recursive call. Specifically, that type is
$\tfunshow {\potp {\tlistshow {0,0,4} \alpha} {3}} {\potp {\tlistshow {24,24,12} \alpha} {3}}$, which in turn requires the type 
$\tfunshow {\potp {\tlistshow {0,0,16} \alpha} {15}} {\potp {\tlistshow {96,96,48} \alpha} {15}}$ to be justified, and so on.

\subsection{Cost-Free Types}
\label{sec:costfree}

The key idea used to solve resource-polymorphic recursion
(and other problems) is \textit{cost-free types} \cite{hoffmann2010amortized}.
Our work is about a new way to
express cost-free types with linear maps.
Here we first explain what cost-free types are, and 
then show how cost-free types provide a means of addressing 
resource-polymorphic recursion.

There is only one difference between a 
cost-free and costful (normal) AARA
type for a function: the cost model considered for the code
being typed. For costful types, the cost-assigning
``tick'' expressions dictate the cost model, while for cost-free types 
``tick'' expressions are ignored. 
Without ``tick'' expressions
consuming any potential to pay for costs, the 
resulting cost-free type only
describes how AARA can
\emph{reallocate}
potential from input to output data structures.
%\twr{I'd recommend moving the paragraph with the analogy to homogeneous solutions to differential equations here.
%That observation provides a framing for the reader trying to understand the rest of the section.}
Thus, cost-free types serve
much the same role as, e.g., homogeneous solutions to differential equations.
Just as any of the solutions of the homogeneous equation can 
be added to the solution of the inhomogeneous equation to 
obtain another solution to the inhomogeneous equation, 
any cost-free type can be added\footnote{
  It is critical
  that cost-free types are \textit{added} and not subtracted.
  Among other reasons, repeatedly subtracting a cost-free type with positive net cost can unsoundly result in an arbitrarily low, even negative, net cost bound.
} 
to a costful 
type to obtain another valid costful type. 
Formally, if the function $f$ can
be given the costful annotation $\vec a$ and
the cost-free $\vec b$, then $\vec a + k\cdot \vec b$ is also a valid costful type for $k\geq 0$.

Cost-free types avoid
the loss of potential by reallocating it instead.
Preventing such loss is critical to maintaining tight cost bounds and good
compositionality.
For example, recall the function \texttt{id} 
from \Cref{fig:idcode}. In \Cref{sec:aara}, \texttt{id} 
had the costful type 
$\tfunshow {\potp {\tlistshow 1 \alpha} {0}}
{\potp {\tlistshow 0 \alpha} {0}}$ using 
linear potential, but \texttt{id} also could be given
other costful types. Having multiple such types is necessary if
one wants to find a good type for a composition like \texttt{id} $\circ$ 
\texttt{id}
$: \tfunshow {\potp {\tlistshow 2 \alpha} {0}} {\potp {\tlistshow 0 \alpha} {0}}$, which requires the additional costful typing \texttt{id}
$: \tfunshow {\potp {\tlistshow 2 \alpha} {0}} {\potp {\tlistshow 1 \alpha} 0}$ for the righthand instance of \texttt{id}. This new
costful type can be obtained from the first by
adding the cost-free
type $\tfunshow {\potp {\tlistshow 1 \alpha} {0}} {\potp {\tlistshow 1 \alpha} 0}$, indicating extra input potential can be reallocated to
the output at a 1:1 rate.

Resource-polymorphic recursion is just a pathological
example of having excess input potential
because each recursive call is applied to an input typed with 
larger (excess) annotations.
Thus, cost-free
types provide a way of approaching that
infinite regress. 
To exemplify how cost-free types
arise, consider \texttt{half} with 
polynomial potential from \Cref{sec:rpr}.
The main issue is at the recursive call, where we want the type $\tfunshow {\potp {\tlistshow {1,2} \alpha} 1} {\potp {\tlistshow {4,5} \alpha} 1}$, but our assumed type for \texttt{half} is $\tfunshow {\potp {\tlistshow {1,0} \alpha} 0} {\potp {\tlistshow {4,1} \alpha} 0}$.
If the ``difference'' between these types is a cost-free
type for \texttt{half}, then the overall typing would be justified.
That difference corresponds to the type
$\tfunshow 
{\potp {\tlistshow {0,2} \alpha} 1} {\potp {\tlistshow {0,4} \alpha} 1}$,
which is a semantically valid cost-free type\footnote{
  In fact, this type is also syntactically justifiable via the basic typing rules without resource-polymorphic recursion.
}
because $2\cdot n + 1 \geq 4 \cdot \lfloor \frac n 2 \rfloor + 1$.
%\twr{What do you mean by ``naively''?
%If you mean ``succeeds using the syntactic type rules'' then state it that way.
%``Naive'' is a word with negative overtones, so ``succeeds naively'' does not come across as a positive.
%}
Similarly, in the exponential case for \texttt{half}, 
the
difference would be
$\tfunshow 
{\potp {\tlistshow {0,0,3} \alpha} 3} {\potp {\tlistshow {18,18,9} \alpha} 3}$, and this type is a semantically valid cost-free type because 
$3 \stirling {2n+1} 2 + 3 = 18\stirling {n+1} 4 + 18\stirling {n+1} 3 + 9\stirling {n+1} 2 + 3$.

\subsection{The Existing Cost-Free Typing Method}
\label{sec:old}

The existing method of inference with cost-free types stems from 
Hoffmann and Hofmann \cite{hoffmann2010amortized}. 
However, the version given there is not general enough 
to handle many high-degree examples. Improvements were made
in later work \cite{HoffmannAH12}, 
and the full picture was never published. 
Here we
sketch the full picture to highlight the innovations of the new cost-free typing method (\Cref{sec:new}).

The basic idea behind the existing inference method 
is to
use brute force
to infer a new cost-free type for the excess potential
present at each and every function call site. 
Recall that annotations in AARA are initially symbolic during
type inference and are only concretized \textit{after} solving
a linear program. Thus, this process of retyping amounts
to assuming that the annotations at \textit{every} function call
are of the form $\vec a + \vec b$, where $\vec a$ matches the costful
typing of the function argument, and $\vec b$ is some amount of 
excess.
Then the function is retyped
using the cost-free typing rules
to  accept an input annotated $\vec b$. Even if it would happen to 
be the trivial case that the excess $\vec b$ is the zero vector,
AARA would not be able to take advantage of this fact until 
\textit{after} the system of linear inequalities is generated and
solved. Thus, the types of functions with trivial cost behavior 
can be equally expensive to infer as those with nontrivial
resource-polymorphic recursion.

However, brute-force retyping alone is not sufficient
for typing recursive functions. Retyping any function at its
recursive call will cause the process to loop.
To avoid looping indefinitely, the existing inference method 
introduces an additional assumption about the structure of 
potential annotations at recursive calls. Specifically,
it is assumed that the highest-degree annotation 
in the excess vector $\vec b$ is 0. This assumption 
reduces
the degree of potential represented by the annotation
vectors at each iteration until 
the base case of linear excess potential.
Under the given assumption, the linear case's
recursive calls to a function are assumed to precisely match the assumed
type of that function,\footnote{
    That is, the type matches \textit{modulo constant potential}. 
    AARA can already freely pass extra constant potential
    through functions, so constant potential can
    be excluded from the concerns of resource-polymorphic recursion.
}
and therefore exhibit no resource-polymorphic recursion. 
This process mirrors how repeatedly taking the (discrete)
derivative of a \textit{polynomial} eventually yields 
the constant 0 function.

We already saw a resource-polymorphic typing for which this algorithm
works perfectly
in \Cref{sec:costfree}. At \texttt{half}'s recursive call, the needed cost-free type was $\tfunshow {\potp {\tlistshow {0,2} \alpha} 1} {\potp {\tlistshow {0,4} \alpha} 1}$, 
which indeed has a 0 in the first position.
This type
only uses linear potential and does not require any resource-polymorphic recursion to
typecheck successfully.
In this case, the key heuristic used in the existing inference method allows the algorithm to converge to the desired solution.
Using higher-degree potential would require more iterations to converge.
% Thus, the key assumption of the existing inference method coincides with the desired solution, so the algorithm converges to that solution as desired.

We now give the inference algorithm
more formally. 
Let $D_{\mathit{max}}$ be the greatest degree of polynomial
potential being considered
and $\vec a$ be the annotation of the type 
of some recursive function \texttt{f}. 
The linear program generated by the existing 
inference algorithm
is constructed
as follows:

\begin{enumerate}
    \item If only linear potential is being used in $\vec a$,
    assume %no resource-polymorphic recursion happens so
    that $\vec a$ is also used for \texttt{f}'s recursive
    call.

    \item If the degree of potential being used in $\vec a$
    is greater than one, assume a cost-free type 
    annotated $\vec b$ of degree $D_{\mathit{max}}-1$ is needed at the recursive call. Attempt to type \texttt{f}
    using the cost-free annotation $\vec b$, then use 
    the annotation $\vec a + \vec b$ at \texttt{f}'s recursive
    call.
\end{enumerate}

In step 2, the annotation at the call is divided
into $\vec a + \vec b$ so that the ``extra'' potential $\vec b$ 
is retyped in a cost-free manner. Note how
step 2 will be repeated for each internal recursive call to $f$
and for each degree of polynomial
potential being used to type $f$.
Thus, if $f$ has $b$ recursive calls and potential 
is of degree $d$,
each external call to $f$ results in retyping $f$ approximately 
$b^d$ times. 
Every retyping of $f$ entails 
retyping each helper function called in
$f$, resulting in a combinatorially explosive cascade of retyping.
We revisit this cost profile quantitatively 
in \Cref{sec:exp,sec:further}.

While this algorithm works for \texttt{half} and many
other functions, it is somewhat lucky
that it finds any solutions at all.
This circumstance is because the key assumption that the algorithm
relies on is not guaranteed to be true. Rather, the assumption is just
\textit{often} true when dealing with simple pattern matching
and polynomial potential because Pascal's identity does
not change the coefficient of the highest degree binomial coefficient.
Thus, that coefficient gets cancelled out when expressing
a cost-free type as the difference of types.

Using polynomial potential, a counterexample to the assumption
can be found in \texttt{round} in \Cref{fig:round}.
This function computes a list of length 
equal to the largest number of the form $2^x-1$ 
that is less than or equal to
the input list's length. The code makes use of \texttt{half}
in addition to pattern matching to shrink the size of the
list before the recursive call,
which increases the potential density. As
a result, using
only linear potential to type function \texttt{round}
as $\tfunshow {\potp {\tlistshow a \alpha } 0}{\potp {\tlistshow a \alpha} 0}$ 
requires a 
type of 
$\tfunshow {\potp {\tlistshow {2a} \alpha } 1}{\potp {\tlistshow {2a} \alpha } 1}$ 
at the recursive call.
The linear potential increases by a factor of 
two with each 
round, 
violating the assumption needed for convergence.

While \texttt{round} is somewhat contrived, one
does not need to look far
to find counterexamples when using 
non-polynomial cost-bound templates. Unlike the annotation for the
highest-degree polynomial potential, the annotation for 
the highest-base exponential potential rarely remains constant, so
it cannot usually be cancelled out by taking differences.
This variation is caused by the linear recurrence for Stirling numbers 
$\stirling {n+1} {k+1} = (k+1)\stirling {n} {k+1} + \stirling {n} {k}$,
which includes the scalar $k+1$, whereas Pascal's identity does not.
For \texttt{half} with exponential potential,
the different form of the recurrence
causes a scaling factor of 4 between the type of $\tfunshow {\potp {\tlistshow {0,0,1} \alpha} 0} {\potp {\tlistshow {6,6,3} \alpha} 0}$ and the type needed at the recursive call of $\tfunshow {\potp {\tlistshow {0,0,4} \alpha} 3} {\potp {\tlistshow {24,24,12} \alpha} 3}$.
Thus, the existing approach cannot find good exponential-cost bounds for \texttt{half}.

\begin{figure}
  \begin{subfigure}[t]{0.4\textwidth}
  \begin{lstlisting}[xleftmargin=4.0ex]
fun dbl lst = 
  case lst of 
  | [] -> [] 
  | x::xs -> 
    x::x::(dbl xs)
  \end{lstlisting}
  \caption{}
  \label{fig:dbl}
  \end{subfigure}
  \begin{subfigure}[t]{0.5\textwidth}
  \begin{lstlisting}[xleftmargin=4.0ex]
fun round lst = 
  case lst of 
  | [] -> [] 
  | x::xs -> 
    x::dbl(round(half xs))
  \end{lstlisting}
  \caption{}
  \label{fig:round}
  \end{subfigure}
  \vspace{-0.5cm}
  \caption{Code for \texttt{dbl} and \texttt{round}}
  \label{fig:code}
\end{figure}

\subsection{The New Cost-Free Typing Method}
\label{sec:new}

The method described in the present paper
follows from the observation that the cost-free types may often all be summarized with a single linear map/matrix $M$.
In the case for \texttt{half} with polynomial or 
exponential potential, such matrices would 
be the following:
\[
polynomial:
\left (
\begin{smallmatrix}
    4 & 0 & 0 \\
    1 & 2 & 0 \\
     0 & 0 & 1 \\
\end{smallmatrix} \right )
\;\;\;\;
exponential:
\left (
\begin{smallmatrix}
    \sfrac {505} {12} & \sfrac {206} 3 & 6 & 0 \\
    10 & 22 & 6 & 0 \\
     0 & 1 & 3 &0 \\
     0 & 0 & 0 & 1 \\
\end{smallmatrix} \right )
\]

By associating
$\potp {\tlistshow {a,b} \alpha} c$ to the vector 
$\left ( \begin{smallmatrix}
  a \\ b \\ c
\end{smallmatrix} \right )$, 
one finds that, e.g., the polynomial map $M$ is consistent with each of 
the polynomial cost-free types found for \texttt{half} in \Cref{sec:rpr}. 
For instance, the types
$\tfunshow {\potp {\tlistshow {1,2} \alpha} 1} 
{\potp {\tlistshow {4,5} \alpha} 1}$
and $\tfunshow {\potp {\tlistshow {0,2} \alpha} 1} {\potp {\tlistshow {0,4} \alpha} 1}$ are encoded by the mappings
\[
\left (
\begin{smallmatrix}
    4 & 0 & 0 \\
    1 & 2 & 0 \\
     0 & 0 & 1 \\
\end{smallmatrix}\right ) 
\left (
\begin{smallmatrix}
  1 \\ 2 \\ 1
\end{smallmatrix} \right )
=
\left (
\begin{smallmatrix}
  4 \\ 5 \\ 1
\end{smallmatrix} \right )
\qquad\qquad
\left (
\begin{smallmatrix}
    4 & 0 & 0 \\
    1 & 2 & 0 \\
     0 & 0 & 1 \\
\end{smallmatrix}\right ) 
\left (
\begin{smallmatrix}
  0 \\ 2 \\ 1
\end{smallmatrix} \right )
=
\left (
\begin{smallmatrix}
  0 \\ 4 \\ 1
\end{smallmatrix} \right )
\]

Crucially, these matrices
are invariant across recursive calls, which
avoids infinite regress. And unlike
the existing cost-free typing algorithm,
inferring these matrices does not require additional
typings of \texttt{half}, so 
type checking/inference with them can be made efficient
(\Cref{sec:aut,sec:exp,sec:further}).

In \Cref{sec:ty}, we define a declarative
type system to justify matrices like these,
and in \Cref{sec:sound} we prove 
that
their use is
sound. While there are many subtleties,
the thrust of our approach is to treat potential transformations
algebraically
and identify matrices that
satisfy certain inequalities. For example, the polynomial matrix $M$
for \texttt{half} satisfies the following inequalities\footnote{
  To clarify notation, $*$ stands in for an arbitrary choice of number throughout this work, so that, e.g., $* \cdot 0 = 0$
  and $*+*=*$.
}
which correspond to the action that \texttt{half} takes on inputs of length 0, 1, or greater, respectively.
\begin{mathpar}
M \leq 
\left (
\begin{smallmatrix}
    * & * & * \\
    * & * & * \\
     0 & 0 & 1 \\
\end{smallmatrix}
\right )

M \leq 
\left (
\begin{smallmatrix}
    * & * & * \\
    * & * & * \\
     0 & 0 & 1 \\
\end{smallmatrix}
\right )
\left (
\begin{smallmatrix}
    1 & 0 & 0 \\
    1 & 1 & 0 \\
     0 & 1 & 1 \\
\end{smallmatrix}
\right )

M \leq 
\left (
\begin{smallmatrix}
    1 & 0 & 0 \\
    1 & 1 & 0 \\
     0 & 1 & 1 \\
\end{smallmatrix}
\right ) ^{-1}
M
   \left (
\begin{smallmatrix}
    1 & 0 & 0 \\
    1 & 1 & 0 \\
     0 & 1 & 1 \\
\end{smallmatrix}
\right ) ^2

\end{mathpar}

Similarly, the exponential matrix satisfies the following inequalities: 
\begin{mathpar}
M \leq 
\left (
\begin{smallmatrix}
    * & * & * & * \\
    * & * & * & * \\
    * & * & * & * \\
     0 & 0 & 0 & 1 \\
\end{smallmatrix}
\right )

M \leq 
\left (
\begin{smallmatrix}
    * & * & * & * \\
    * & * & * & * \\
    * & * & * & * \\
     0 & 0 & 0 & 1 \\
\end{smallmatrix}
\right )
\left (
\begin{smallmatrix}
    4 & 0 & 0 &0  \\
    1 & 3 & 0 &0 \\
     0 & 1 & 2 &0 \\
     0 & 0 & 1 & 1
\end{smallmatrix}
\right )

M \leq 
\left (
\begin{smallmatrix}
    4 & 0 & 0 &0  \\
    1 & 3 & 0 &0 \\
     0 & 1 & 2 &0 \\
     0 & 0 & 1 & 1
\end{smallmatrix}
\right ) ^{-1}
M
   \left (
\begin{smallmatrix}
    4 & 0 & 0 &0  \\
    1 & 3 & 0 &0 \\
     0 & 1 & 2 &0 \\
     0 & 0 & 1 & 1
\end{smallmatrix}
\right ) ^2

\end{mathpar}

To ensure soundness, the full system of \Cref{sec:ty} adds a few extra inequalities beyond these.

We also preemptively note that this new method of
cost-free typing
is not a strict improvement over the existing method. 
While type inference
is usually faster with our new method,
our new method cannot type all functions as well as the previous 
approach. For example, because our method represents
cost-free types with 
linear maps, it cannot express reallocation of potential based on
a nonlinear function like \texttt{min}, nor can it express multiple choices
of potential reallocation. We revisit such limitations in \Cref{sec:lim}.

Because each cost-free typing approach has its pros and cons, we 
recommend that an implementation of cost-free typing should employ a hybrid approach.
Specifically, we recommend to first attempt
our more efficient approach and then default to the
previous approach should our approach fail. This
hybridization would allow many simpler functions to be handled
quickly while still providing the robust coverage of the previous
approach. In particular, if a helper function is amenable to
our new approach, then that function need not be retyped
repeatedly when retyping its callers, cutting short the 
expensive cascade of retyping described in \Cref{sec:old}.

\section{Type System}
\label{sec:ty}

In this section, we define our cost-free
type system, as well as additional formalisms
the type system interacts with. These formalisms include the 
language analyzed by the system, as well as the
types, potential, and primitive maps the system uses. 
To focus on the salient parts of our work, 
we only present the parts of these 
formalisms associated with Booleans, lists, and functions.

\subsection{Language} 
Our type system analyzes a call-by-value functional language with 
commonplace features like pattern matching.
We simplify our language
by first syntactically transforming expressions into 
\textit{let-normal form}, 
a restriction of A-normal 
form \cite{sabry1992reasoning}
wherein every subexpression is
bound to a variable. (A-normal form
usually allows subexpressions to be variables \textit{or} constants.) We then 
give our expressions by the following grammar
where $f$ and $x$ represent variable names:
\begin{align*}
    e ::=\; & \etrue \mid \efalse \mid 
    \econd x {e_1} {e_2}
    \mid \enil
    \mid \econs {x_h}  {x_t}
    \mid (\ecasel {x_\ell} {e_1} {x_h} {x_t} {e_2})
    \\
    &\mid x \mid \elet x {e_1} {e_2}
    \mid \efun f x e \mid \eapp {f} x
\end{align*}

Note that our language does not include the ``tick''
expressions mentioned in \cref{sec:aara}. This is 
because cost-free types ignore ticks, so they are not 
needed.

Our language's values are then 
given by the following
grammar:
\begin{align*}
    v ::=\; & \etrue \mid \efalse 
    \mid \enil
    \mid \econs {v_h}  {v_t}
    \mid \vclosure V f x e
\end{align*}
where $\vclosure V f x e$ is a function closure for the recursive
function $f$ that takes argument $x$, has body $e$, and captures
the value context $V$. 

We address what it means for a value to be well-formed in \Cref{sec:sound}.

\subsection{Types}
\label{subsec:types}

The types considered in this work are given by the following grammar: 
\[ \tau,\sigma ::= \; \tbool \mid \tlist {\tau} \mid 
    \tfun \tau \sigma M,
\]
where $L(\cdot)$ is the list type constructor.

These types differ from standard AARA types in a few 
key ways.
First, lists types are no longer annotated---our 
type system does not concern itself with any \emph{specific} annotation,
but instead captures how annotations are \emph{transformed}.
This focus is reflected in the second key way that 
this type system differs: function types are annotated 
with a linear map (matrix) $M$. This map represents how 
annotations are transformed by
an application of the function.
Roughly speaking, $M$'s domain is the annotation vector for
the function's
argument, and $M$'s range is the annotation vector for the
function's return.

\begin{figure}
\begin{small}
\begin{mathpar}
\inds{\tau} = \begin{cases}
    \{\ideg n \mid 1 \leq n \leq D_{\mathit{max}} \} & \tau = \tlist {\sigma},\, \mathit{polynomial\; potential}
    \\
    \{\ibase n \mid 2 \leq n \leq B_{\mathit{max}} \} & \tau = \tlist {\sigma},\, \mathit{exponential\; potential}
    \\
    \emptyset & \mathit{otherwise}
\end{cases}

\inds{\Gamma} = (\bigcup_{x:\tau \in \Gamma} x \idot \inds \tau) \cup \{\iconst\}

\end{mathpar}
\end{small}
\caption{Indices}
\label{fig:inds}
\end{figure}

We use a special set of indices given in \Cref{fig:inds} 
to formalize the annotation vectors
described in \Cref{sec:aara}.
These indices allow us to annotate whole 
type contexts and to define the domain 
and range of our linear maps. 
We use the 
the symbol ``$\idot$'' as string concatenation for indices, 
which is evocative of its use
in $x \idot \inds \tau$ as a
field access.
We also simplify notation
by letting such concatenation distribute over sets of indices.

Each index in $\inds \tau$ 
is a string
representing a slot where an annotation
would
go on the type $\tau$ in standard AARA. 
Thus, annotation vectors are just (finite) maps
from indices to numbers.
In our setting, we consider only
so-called ``shallow''
indices,\footnote{
  Standard AARA includes deep indices, like on the inner lists in types such as $\tlist {\tlist {\tau}}$.
  We ignore such inner lists in our setting because they complicate matters a lot and are virtually never relevant for resource-polymorphic recursion.
}
so the list type only has indices for the degrees ($\ideg n$) or bases ($\ibase n$) of potential being considered, up to some maximum.
These $\ideg n, \ibase n$ indices correspond to the $n^{\mathit{th}}$  element of the annotation vectors discussed in \Cref{sec:aara}.
No other type has any indices.
Across an entire typing 
context, we also prefix each index by the name of the 
variable it belongs to. Thus, the index 
$\mathtt x \idot \ideg 2$ refers to the
quadratic-potential annotation of the list $\mathtt x$.
Finally, type contexts have the additional constant
index $\iconst$, which refers to the amount of constant potential available to pay for costs.
We
subscript symbolic
vectors with these indices to
pick out annotations, so that,
e.g., $\vec p_{\mathtt{x}\idot \ideg 2}$ gives the
quadratic potential assigned to $\mathtt{x}$ by $\vec p$.

\begin{example} \label{ex:anno}
The following explicitly-indexed
column
vectors are the polynomial annotation of
the type context $\potp {x:\tlistshow {3,1,4} \alpha, b:\tbool} 5$ and the exponential annotation of $\potp {y:\tlistshow {2,0,6} \alpha, b:\tbool} 7$:
\begin{mathpar}
\textrm{polynomial:}
\begin{scriptsize}
\begin{blockarray}{cc}
  \begin{block}{c(c)}
        x \idot \ideg 3   & 3 \\
        x \idot \ideg 2   & 1 \\
        x \idot \ideg 1   & 4 \\
     \iconst   & 5 \\
  \end{block}
  \end{blockarray}
\end{scriptsize}

\textrm{exponential:}
\begin{scriptsize}
  \begin{blockarray}{cc}
  \begin{block}{c(c)}
        y \idot \ibase  4   & 2 \\
        y \idot \ibase 3   & 0 \\
        y \idot \ibase 2   & 6 \\
     \iconst   & 7 \\
  \end{block}
  \end{blockarray}
  \end{scriptsize}
\end{mathpar}

Thus, the polynomial annotation denotes the potential 
function
$\lambda n . 3\binom{n}{3} + \binom{n}{2} + 4\binom{n}{1} + 5$, 
where $n$ is a formal parameter for the length of list $x$. Likewise, the exponential annotation denotes
$ \lambda m . 2\stirling{m+1}{4} + 6\stirling{m+1}{2} + 7$, where $m$ is the formal parameter for the length of list $y$.
\end{example}

To describe the domain and range of
the linear map associated with a function type,
we must also introduce the special names $\iarg$ and $\iret$ for the otherwise-unnamed argument and return
values.
Thus, for $\tfun \tau \sigma M$, 
the map $M$ maps annotation vectors indexed by $\iarg \idot \inds \tau \cup \{\iconst \}$ to annotation vectors indexed by $\iret \idot \inds \sigma \cup \{\iconst \}$. 
We implicitly extend these maps and others to 
act on larger domains simply by treating $M$ as the direct 
sum $M \oplus I$, which acts as
the identity over indices neglected by $M$'s domain.
This interpretation means that $M$ maps $\iarg \idot \inds \tau$ to the zero vector, leaving function arguments with no
potential after application.\footnote{
  This decision to zero the argument potential is made for simplicity of 
  presentation. Some more interesting ways of restoring potential
  to the argument, 
  like remainder contexts \cite{kahn2021automatic,grosen2023automatic}, 
  can be implemented without issue.
}

\begin{figure}
\begin{small}
\begin{mathpar}
\pot v {\tau} {\vec p} = \begin{cases}
    \sum_{k=1}^{D_{\mathit{max}}} \vec p_{\ideg k} \cdot \binom {|v|} k & \tau = \tlist {\sigma},\, \mathit{polynomial}
    \\
    \sum_{k=2}^{B_{\mathit{max}}} \vec p_{\ibase k} \cdot \stirling {|v|+1} k & \tau = \tlist {\sigma},\, \mathit{exponential}
    \\
    0 & \mathit{otherwise}
    \end{cases}

    \pot {V} \Gamma {\vec p} = \vec p_\iconst + \sum_{x:\tau \in \Gamma} \pot {V(x)} \tau {\lambda i.\,\vec p_{x \idot i}}
\end{mathpar}
\end{small}
\vspace{-.5cm}
\caption{
potential
}
\label{fig:pot}
\end{figure}

\subsection{Potential}
We use $\pot v \tau {\vec p}$ to formally denote
the amount of potential stored on 
a value $v$ of type $\tau$ annotated by $\vec p$.
The only type for which this notation assigns non-zero
potential is for lists. We give the formal
definition of potential assignment for both
values and whole contexts in \cref{fig:pot},
where we also extend the notation over contexts.
This definition is consistent with the potential functions
found in \Cref{ex:anno}.

\subsection{Primitive Maps}

\begin{figure}
    \begin{scriptsize}
    \begin{mathpar}
    \mmove x y = 
    \begin{blockarray}{ccc}
  {} &  x \idot \ideg 2 & x \idot\ideg1  \\
  \begin{block}{c(cc)}
    x \idot \ideg 2   & 0 & 0  \\
    x \idot \ideg 1   & 0 & 0  \\
     y \idot \ideg 2 & 1 & 0  \\
     y \idot \ideg 1  & 0 & 1  \\
  \end{block}
  \end{blockarray}
  
  \mshift x y = 
    \begin{blockarray}{cccc}
  {} & x \idot \ideg 2 & x \idot\ideg1 & \iconst   \\
  \begin{block}{c(ccc)}
        x \idot \ideg 2  & 0 & 0 & 0 \\
        x \idot \ideg 1   & 0 & 0 & 0 \\
     y \idot \ideg 2 & 1 & 0 & 0 \\
     y \idot \ideg 1  & 1 & 1 & 0 \\
     \iconst & 0 & 1 & 1 \\
  \end{block}
  \end{blockarray}
  
  \munshift x y = 
    \begin{blockarray}{cccc}
  {} &  x \idot \ideg 2 & x \idot\ideg1 & \iconst   \\
  \begin{block}{c(ccc)}
        x \idot \ideg 2  & 0 & 0 & 0 \\
        x \idot \ideg 1   & 0 & 0 & 0 \\
     y \idot \ideg 2  & 1 & 0 & 0 \\
     y \idot \ideg 1  & -1 & 1 & 0 \\
     \iconst  & 1 & -1 & 1 \\
  \end{block}
  \end{blockarray}

 \vspace{-.3cm}
 
  \mzero x = \begin{blockarray}{cccccc}
  {} &  x \idot \ideg 2  & x \idot \ideg 1 &  y \idot \ideg 2  & y \idot  \ideg 1 
  &\iconst   \\
  \begin{block}{c(ccccc)}
        x \idot \ideg 2   & 0 & 0 & 0 & 0 & 0  \\
        x \idot \ideg 1   & 0 & 0 & 0 & 0 & 0\\
        y \idot \ideg 2   & 0 & 0 & 1 & 0 & 0\\
        y \idot \ideg 2   & 0 & 0 & 0 & 1 & 0\\
        \iconst   & 0 & 0 & 0 & 0 & 1\\
  \end{block}
  \end{blockarray}

    \mnil x = \begin{blockarray}{cccc}
  {} &  x \idot \ideg 2 & x \idot\ideg1 & \iconst   \\
  \begin{block}{c(ccc)}
    x \idot \ideg 2 & * & * & *\\
    x \idot \ideg 1 & * & * & *\\
     \iconst  & 0 & 0 & 1 \\
  \end{block}
  \end{blockarray} 
    \end{mathpar}
    \end{scriptsize}
     \vspace{-.7cm}
    \caption{Selected primitive maps with explicit indices for polynomial potential up to degree 2
    }
    \label{fig:prim}
\end{figure}

\begin{figure}
\begin{scriptsize}
\begin{mathpar}
\mshift x y = 
    \begin{blockarray}{cccc}
  {} &  x \idot\ibase 3 & x \idot\ibase 2& \iconst   \\
  \begin{block}{c(ccc)}
        x \idot \ibase 3   & 0 & 0 & 0\\
        x \idot \ibase 2    & 0 & 0 & 0 \\
        y \idot \ibase 3   & 3 & 0 & 0\\
        y \idot \ibase 2    & 1 & 2 & 0 \\
     \iconst  & 0& 1 & 1 \\
  \end{block}
  \end{blockarray}
  
  \munshift x y = 
    \begin{blockarray}{cccc}
  {}  & x \idot\ibase 3 & x \idot\ibase 2& \iconst   \\
  \begin{block}{c(ccc)}
        x \idot \ibase 3   & 0 & 0 & 0\\
        x \idot \ibase 2    & 0 & 0 & 0 \\
        y \idot \ibase 3   & \sfrac 1 3 & 0 & 0\\
        y \idot \ibase 2   & -\sfrac 1 6 & \sfrac 1 2 & 0 \\
     \iconst  & \sfrac 1 6 & -\sfrac 1  2& 1 \\
  \end{block}
  \end{blockarray}
  \end{mathpar}
  \end{scriptsize}
  \vspace{-.7cm}
    \caption{Primitive shift maps with
    explicit indices for exponential 
    potential up to base 3}
    \label{fig:expprim}
\end{figure}

To simplify the presentation of matrices in our type system,
we build more complex maps out
of a primitive set of maps.
We exemplify such maps
in \Cref{fig:prim,fig:expprim}.
These primitive maps correspond to the annotation
manipulations induced by basic evaluation steps and are
often representable as simple matrices. The matrix
dimensions (and the
matrices for ``shift'' and ``unshift''
in particular) are
implicitly parameterized by the kinds of annotations in 
use, i.e., polynomial or exponential, and maximum degree 
($D_{\mathit{max}}$) or base ($B_{\mathit{max}}$) annotation. 
Furthermore, we treat
every matrix 
$M$ implicitly as the direct sum $M\oplus I$, 
so that every matrix acts as the identity on
indices it would otherwise neglect. 
We explain each primitive map in detail as follows:

\begin{itemize}
    \item $\mi$ -- This map is just the identity map, and is
    used when the potential-carrying variables remain unchanged {by a program action}. (We do not show $\mi$ in \Cref{fig:prim,fig:expprim}.)

    \item $\mmove x y$ -- This map moves
    all potential on the variable 
    $x$ to $y$. 
    Each value at index 
    $x \idot i$ ends up at index $y \idot i$
    after application. We assume that relabeling only 
    occurs with $y$ an unused name or a variable with zero potential so that we avoid any name collisions.

    \item $\mshift x y$ -- This ``shift'' map
    applies a transformation to the list indices of 
    $x$ and puts the results under the new name $y$,
    adjusting the constant amount as appropriate.
    This matrix typically appears at a list pattern-matching operation.
    The particular transformation applied is specific to the
    type of potential used
    (i.e., polynomial vs.\ exponential),
    and must relate the annotations before and after
    a list pattern-matching operation.
    Depending on how this relation is defined, one can obtain polynomial or exponential potential (compare \Cref{fig:prim,fig:expprim}).
    For example, the version for polynomial potential uses Pascal's identity. 
    As with 
    $\mmove x y$, we assume that the name $y$ is unused.

    \item $\munshift x y$ -- This ``unshift'' map acts as
    the inverse of 
    $\mshift y x$. This map adjusts a list's annotations 
    when an element is added to its front.
    This matrix and $\mshift y x$ are the only primitive maps 
    that need to change to convert the type system from capturing polynomial potential to capturing exponential potential.

    \item $\mnil x$ -- This map is a special havoc
    operation used when instantiating or pattern-matching
    the empty list. We represent
    this havocking by putting the symbol $*$ across
    each row $x \idot i$. In this work, $*$ stands for
    an arbitrary choice of number, so that, e.g., $*\cdot 0 = 0$
    and $* + r = * \cdot * = *$.

    \item $\mproj {x}$ -- This map
    projects onto
    indices of the form $x\idot i$ if $x$ is a variable name, or the index $\iconst$ if $x$ is $\iconst$. This projection is accomplished by
    zeroing all other indices.
    We extend the notation to sets so that $\mproj {\{x,y \}}$
    projects onto both $x$ and $y$.

    \hspace{1.5ex}
    We write $\mproj {\neg x}$ for the complement projection that projects \emph{away} indices of the form $x\idot i$.
    The map $\mproj {\neg x}$ is used when variable $x$ leaves or 
    (freshly) enters a scope.

    \iffalse
    \item $\mzero x$ -- This map is the zero map on indices of $x$.
    This matrix corresponds to introduction of $x$ into a scope and ensures that $x$ holds no potential.
    \fi
\end{itemize}

By construction, $\pot {(x\mapsto v_h::v_t)} {(x:\tlist \tau)}
{\vec p} = \pot {(y\mapsto v_t)} {(y:\tlist \tau)} {\mshift x y \cdot \vec p}$
for any annotation vector $\vec p$. 
That is, the shift map perfectly
conserves potential. 
Similarly, the map $\mnil x$ preserves potential:
$\pot {(x\mapsto \enil )} {(x:\tlist \tau)}
{\vec p} = \pot {(x\mapsto \enil )} {(x:\tlist \tau)} {\mnil x\cdot \vec p}
= 0$.
Each of the primitive maps (except for 
$\mproj x$)
is designed to be conservative on its
corresponding operations, 
even given an annotation vector $\vec p$ with negative entries.

\begin{example}
In \Cref{sec:rpr}, the last step of \texttt{half} creates the return value by attaching \texttt{x1} to the front of list \texttt{tmp}.
Our informal reasoning about this step transformed potential using the identity
$4\binom m 2 + 5 \binom m 1 + 1 = 4 \binom {m+1} 2 + 1 \binom {m+1} 1$.
Using the unshift map $\munshift {\mathtt{tmp}} \iret$, this step is formalized as follows:
\begin{scriptsize}
\[
\munshift {\mathtt{tmp}} \iret \cdot 
\begin{blockarray}{cc}
  \begin{block}{c(c)}
        \mathtt{tmp} \idot \ideg 2   & 4 \\
        \mathtt{tmp} \idot \ideg 1   & 5 \\
     \iconst   & 1 \\
  \end{block}
  \end{blockarray}
  \;\;\;\;=\;\;\;\;
  \begin{blockarray}{cccc}
  {} & \mathtt{tmp}\idot \ideg 2 & \mathtt{tmp}\idot \ideg 1  & \iconst \\
  \begin{block}{c(ccc)}
        \mathtt{tmp} \idot \ideg 2   & 0 & 0 & 0\\
        \mathtt{tmp} \idot \ideg 1   & 0 & 0 & 0 \\
        \iret \idot \ideg 2 & 1 & 0 & 0 \\
        \iret \idot \ideg 1 & -1 & 1 & 0 \\
        \iconst   & 1 & -1 & 1 \\
  \end{block}
  \end{blockarray}
  \cdot 
  \begin{blockarray}{cc}
  \begin{block}{c(c)}
        \mathtt{tmp} \idot \ideg 2   & 4 \\
        \mathtt{tmp} \idot \ideg 1   & 5 \\
     \iconst   & 1 \\
  \end{block}
  \end{blockarray}
   \;\;\;\;=\;\;\;\;
  \begin{blockarray}{cc}
  \begin{block}{c(c)}
        \mathtt{tmp} \idot \ideg 2   & 0 \\
        \mathtt{tmp} \idot \ideg 1   & 0 \\
        \iret\idot \ideg 2   & 4 \\
        \iret \idot \ideg 1   & 1 \\
     \iconst   & 0 \\
  \end{block}
  \end{blockarray}
\]
\end{scriptsize}
\end{example}

\subsection{Typing Rules}
\label{sec:trules}

\begin{figure}
\begin{scriptsize}
\begin{mathpar}
\inferrule[M-Fun]{
    \tjudgecf {\Gamma, f: \tfun \tau \sigma M, x: \tau} e \sigma {\mathcal S} {\mathfrak C}
    \\
     \mproj  {\{x, \iret,\iconst \} } \cdot \mathcal S \cdot \mzero {\dom{\Gamma}} 
     \geq \mproj {\{ x, \iret,\iconst \} } \cdot M \cdot \mmove x \iarg 
    \\
    \mproj  {\neg \{x, \iret,\iconst \} } \cdot \mathcal S \cdot \mzero {\dom{\Gamma}} \geq 0
    \\
    \mathfrak C \cdot \mzero {\dom{\Gamma}} \geq 0
}{
    \tjudgecf \Gamma {\efun f x e} {\tfun \tau \sigma M}
    { \{ \mi \}} {\emptyset }
}

\inferrule[M-App]{
}{
    \tjudgecf {\Gamma, f: \tfun \tau \sigma M, x: \tau} {\eapp f x} \sigma { \{ M\cdot \mmove x \iarg \} } {\{ \mproj { \{ x,\iconst \} } \}}
}

\inferrule[M-Var]{
}{
    \tjudgecf {\Gamma, x:\tau} {\evar x} \tau { \{ \mmove x \iret \} }
    {\emptyset}
}

\inferrule[M-Nil]{
}{
    \tjudgecf {\Gamma} {\enil} {\tlist \tau} { \{\mnil \iret \} } \emptyset
}

\inferrule[M-CaseL]{
    \tjudgecf {\Gamma, \ell : \tlist \tau } {e_1} \sigma {\mathcal S} {\mathfrak C }
    \\
    \tjudgecf {\Gamma,  \ell : \tlist \tau, h:\tau, t : \tlist \tau } {e_2} \sigma {\mathcal T}
    {\mathfrak D}
}{
    \tjudgecf {\Gamma, \ell : \tlist \tau } {\ecasel \ell {e_1} h t {e_2}}
    \sigma
    {
        \mathcal S \cdot \mnil \ell 
        \cup \mproj {\neg \{h,t\} } \cdot \mathcal T \cdot \mshift \ell t \cdot \mzero h
    }
    {
    \mathfrak C \cdot \mnil \ell
    \cup 
    (\mproj {\{ h, t\}} \cdot \mathcal T \cup \mathfrak D) \cdot \mshift \ell t\cdot \mzero h
    }
}

\inferrule[M-Cons]{
}{
    \tjudgecf {\Gamma,h:\tau, t : \tlist \tau } {\econs h t} {\tlist \tau } {\{ \munshift t \iret \} } \emptyset
} 

\inferrule[M-Cond]{
    \tjudgecf {\Gamma, b:\tbool } {e_1} \tau {\mathcal S} {\mathfrak C}
    \\
    \tjudgecf {\Gamma, b:\tbool } {e_2} \tau {\mathcal T} {\mathfrak D}
}{
    \tjudgecf {\Gamma, b:\tbool } {\econd b {e_1} {e_2}} \tau {\mathcal S \cup \mathcal T}
    {\mathfrak C \cup \mathfrak D}
}

\inferrule[M-Bool]{
    e \in \{ \etrue, \efalse\}
}{
    \tjudgecf {\Gamma} e {\tbool} {\{ \mi \}} \emptyset
}

\inferrule[M-Let]{
    \tjudgecf {\Gamma} {e_1} \sigma {\mathcal S} {\mathfrak C}
    \\
    \tjudgecf {\Gamma, x:\sigma} {e_2} \tau {\mathcal T} {\mathfrak D}
}{
    \tjudgecf {\Gamma} {\elet x {e_1} {e_2}} \tau 
    {\mproj {\neg x} \cdot \mathcal T \cdot \mmove \iret x \cdot \mathcal S}
    {\mathfrak C \cup (\mproj x \cdot \mathcal T \cup \mathfrak D) \cdot \mmove \iret x \cdot \mathcal S}
}
\end{mathpar}
\end{scriptsize}
\caption{Linear-map-based, cost-free typing rules}
\label{fig:trules}
\end{figure}

While we are only actually interested in the cost-free
types of functions, we must build up to those 
types compositionally by typing all subexpressions.
We formalize typing using the rules in
\Cref{fig:trules}. To ease notation, these rules implicitly
extend matrix operations to work over sets.
In particular, $\mathcal S \geq M$ means
$\forall S \in \mathcal S$, the matrix $S$ is a pointwise upper bound on the matrix $M$.
The rules also use the following typing judgement:
\[
\tjudgecf{\Gamma} e \tau {\mathcal S} {\mathfrak C}
\]
This typing judgment means:
\begin{itemize}
    \item in typing context $\Gamma$, the expression $e$ has type $\tau$
    
    \iffalse
    \item the evaluation of $e$ might manipulate
potential annotations for in-scope variables and the returned value according to any of the linear maps in the finite 
set $\mathcal S$
    \item the evaluation of $e$ might manipulate
potential annotations for variables that fall out of scope and function arguments according to any of the linear maps in the finite 
set $\mathfrak C$
    \fi
    
    \item the evaluation of $e$ might manipulate potential annotations 
    according to any of the linear maps in $\mathcal S$ and 
    $\mathfrak C$, detailed in this section as follows
\end{itemize}  

The basic approach of our cost analysis is to
use our type system to assign two sets of linear maps, $\mathcal S$
and $\mathfrak C$, to 
each program subexpression. (We will refer to these sets as 
$\mathcal S$ and $\mathfrak C$ throughout this work.)
The maps in $\mathcal S$ are what we use to generate the
matrix inequalities exemplified in \Cref{sec:new}, and 
there is one map in $\mathcal S$ for each possible
path through a function's source code to the return. 
Once a function's body is fully typed, these paths are approximated
by the single matrix that annotates the cost-free function type.
The maps in $\mathfrak C$ are used to generate
some additional matrix inequalities
to ensure that potential is properly conserved,
and there is one map
in
$\mathfrak C$ for each
possible path through the program's source to a 
function application or to the end of 
a variable's scope.
The role of $\mathfrak C$ is to address the following problem:
if a variable has negative potential annotations when it falls out of scope, then there would be an overall gain of potential, which could lead to an unsound conclusion.
To avoid this situation, the potential transformations in $\mathfrak{C}$ are used to check certain non-negativity side-conditions.
The role of $\mathfrak C$ is discussed in \Cref{sec:sound}.

One might wonder why a \textit{set} of matrices $\mathcal S$ is tracked,
rather than one single matrix. If a set of matrices can be
approximated by a single matrix for typing a function, why 
can't the same be done to simplify down to a single
matrix at every branch? The answer is simply that such
a design makes the system more difficult to automate.
Approximating the matrices $M$ and $N$ by some $P\leq M,N$
introduces the unknown matrix $P$, which must be solved for.
As can be seen in \Cref{sec:aut}, such extra unknowns can 
introduce nonlinear constraints, rendering automation impractical.
We minimize this problem by tracking \textit{sets}
of matrices in our type rules in \Cref{fig:trules}. 

The typing rules in \Cref{fig:trules} are mostly 
direct encodings of the typing rules of standard AARA. 
Many rules merely
transform $\mathcal S$ via
the primitive map appropriate for the kind of expression being typed.
For example, the typing rule \textit{M-Cons} merely introduces the unshift map.
Rules that introduce new variables also add matrices to $\mathfrak C$ that pick out the annotations of such
new variables before they fall out of scope. For example
\textit{M-Let} adds the matrices $\mproj x \cdot \mathcal T \cdot \mmove \iret x \cdot \mathcal S$ to $\mathfrak C$, which give the annotations of $x$ at the end
of expression $e_2$. Similarly, rule \textit{M-App} also includes
the matrix $\mproj {\{x,\iconst\}}$ in $\mathfrak C$, which gives
the annotation of $x$ (and the free potential) before a function
application.

Among all the typing rules, the function-typing rule \textit{M-Fun} is
distinct. This rule makes use of
three special matrix-inequality premises that no other rule has,
and these premises are used to allow the sound approximation
of sets of linear maps with a single matrix. We discuss soundness
further in \Cref{sec:sound}.
Intuitively, these
premises encode the following additional information,
assuming that annotations are initially non-negative:

\begin{itemize}
    \item $\mproj  {\{x, \iret,\iconst \} } \cdot \mathcal S \cdot \mzero {\dom{\Gamma}} 
     \geq \mproj {\{ x, \iret,\iconst \} } \cdot M \cdot \mmove x \iarg 
    $. \;\; The matrices in $\mathcal S$ yield a pointwise upper bound on
    the matrix $M$.
    \item 
    $\mproj  {\neg \{x, \iret,\iconst \} } \cdot \mathcal S \cdot \mzero {\dom{\Gamma}} \geq 0
    $. \;\; Variables captured in the function's closure do not have negative annotations when the function body finishes evaluating.
    \item $
    \mathfrak C \cdot \mzero {\dom{\Gamma}} \geq 0$.
    \;\; Arguments to functions and variables that fall out of scope in the course of evaluating the function body do not have 
    negative annotations.
\end{itemize}

There is at least one important remaining difference between 
standard cost-free AARA and our type system:
whereas standard AARA allows 
multiple choices for how potential is allocated, our rules 
codify one particular choice. For example, when a variable 
is used twice in standard AARA, potential may be split
any way between the two uses, which is called ``sharing.''
In our cost-free 
rules, all potential is allocated to a variable's first use, leaving
0 potential for
future additional uses.
By fixing this arbitrary sharing decision, our system avoids 
\textit{full} sharing, while still 
remaining able to type many functions of interest.
Supporting full sharing would yield nonlinear
annotation relations and 
thus prevent the analysis from being captured
with simple linear maps and from being automated via a linear program 
(see \Cref{sec:aut,sec:further}). 
 
\begin{example}\label{ex:HalfWithPolynomialPotential}
Consider typing \texttt{half} from \Cref{fig:half}
as $\tfun {\tlist \alpha} {\tlist \alpha} M$ with polynomial
potential. The typing 
rule \textit{M-Fun} can be used to justify
that the following matrix from \Cref{sec:overview} 
works for $M$: 

\begin{scriptsize}
\[
\begin{blockarray}{cccc}
  {} &  \iarg \idot \ideg 2 & \iarg \idot\ideg1 & \iconst   \\
  \begin{block}{c(ccc)}
     \iret\idot \ideg 2 & 4 & 0 & 0 \\
     \iret \idot \ideg 1 & 1 & 2 & 0 \\
     \iconst & 0 & 0 & 1 \\
  \end{block}
  \end{blockarray}
\]
\end{scriptsize}
There are three paths through the body of \texttt{half}, so our
typing rules will generate three matrices in the set $\mathcal S$.
The first two of these paths return the empty list for inputs of 
length zero or one and correspond to the following matrix products
\footnote{
Recall that each matrix $N$ is implicitly extended 
to $N\oplus I$, so matrices of mismatched dimensions can be multiplied.
}
(where $*$ represents havocked matrix entries):
\begin{scriptsize}
\begin{mathpar} 
    \mnil \iret \cdot \mnil {\mathtt{lst}} = 
    \begin{tiny}
\begin{blockarray}{cccc}
  {} &  \mathtt{lst} \idot \ideg 2 & \mathtt{lst}\idot\ideg1 & \iconst   \\
  \begin{block}{c(ccc)}
   \mathtt{lst}\idot \ideg 2 & * & * & * \\
     \mathtt{lst} \idot \ideg 1 & * & * & * \\
     \iret\idot \ideg 2 & * & * & * \\
     \iret \idot \ideg 1 & * & * & * \\
     \iconst & 0 & 0 & 1 \\
  \end{block}
  \end{blockarray}
\end{tiny}

    \mproj {\neg \{\mathtt{xs1,x1}\}} \cdot \mnil \iret \cdot \mnil {\mathtt{xs1}} \cdot \mshift {\mathtt{lst}} {\mathtt{xs1}} \cdot \mzero {\mathtt{x1}}= 
    \begin{tiny}
\begin{blockarray}{cccc}
  {} &  \mathtt{lst} \idot \ideg 2 & \mathtt{lst}\idot\ideg1 & \iconst   \\
  \begin{block}{c(ccc)}
   \mathtt{lst}\idot \ideg 2 & 0 & 0 & 0 \\
     \mathtt{lst} \idot \ideg 1 & 0 & 0 & 0 \\
     \iret\idot \ideg 2 & * & * & * \\
     \iret \idot \ideg 1 & * & * & * \\
     \iconst & 0 & 1 & 1 \\
  \end{block}
  \end{blockarray}
\end{tiny}
\end{mathpar}
\end{scriptsize}

The remaining path through \texttt{half} is that which makes
a recursive call on inputs of length greater than one.
Using our choice of $M$, the typing rules generate the following
matrix product:

\begin{scriptsize}
\begin{mathpar}
\mproj {\neg \{\mathtt{xs1,x1}\}} \cdot \mproj {\neg \{\mathtt{xs2,x2}\}} \cdot \mproj {\neg \mathtt{tmp}} \cdot \munshift {\mathtt{tmp}} \iret \cdot \mmove \iret {\mathtt{tmp}} \cdot M \cdot \mmove {\mathtt{xs2}} \iarg \cdot \mshift {\mathtt{xs1}} {\mathtt{xs2}} \cdot \mzero {\mathtt{x2}} \cdot \mshift {\mathtt{lst}} {\mathtt{xs1}} \cdot \mzero {\mathtt{x1}}
=
\begin{tiny}
\begin{blockarray}{cccc}
  {} &  \mathtt{lst} \idot \ideg 2 & \mathtt{lst}\idot\ideg1 & \iconst   \\
  \begin{block}{c(ccc)}
   \mathtt{lst}\idot \ideg 2 & 0 & 0 & 0 \\
     \mathtt{lst} \idot \ideg 1 & 0 & 0 & 0 \\
     \iret \idot \ideg 2 & 4 & 0 & 0 \\
     \iret \idot \ideg 1 & 1 & 2 & 0 \\
     \iret & 0 & 0 & 1 \\
  \end{block}
  \end{blockarray}
\end{tiny}
\end{mathpar}
\end{scriptsize}

The typing rule $\textit{M-Fun}$ then requires that, for
each of these $S \in \mathcal S$, two sets of inequalities hold.

The first set of inequalities specifies that,
after identifying the
labels \texttt{lst} and $\iarg$, each matrix $S$ 
is a pointwise upper bound on $M$ over the indices for
$\mathtt{lst}$, $\iret,$ and $\iconst$, i.e.,
$\mproj {\{\mathtt{\mathtt{lst},\iret,\iconst}\}}\cdot S \cdot \mzero {\dom {\Gamma}} 
\geq \mproj {\{\mathtt{\mathtt{lst},\iret,\iconst}\}}\cdot M \cdot  \mmove {\mathtt{\mathtt{lst}}} {\iarg}$.
These inequalities simplify to 
those given in \Cref{sec:new}, which $M$ satisfies.

The second set of inequalities specifies
that each matrix $S$ 
avoids negative entries over indices for captured variables, i.e.,
$\mproj {\neg \{\mathtt{\mathtt{lst},\iret,\iconst}\}}\cdot S \cdot \mzero {\dom {\emptyset}} \geq 0$. 
These second inequalities are trivial because there are no
captured variables in \texttt{half}, so the inequalities are over matrices
of dimension 0.

The typing rules also generate the following matrices in $\mathfrak C$ for variables that leave scope:
\begin{scriptsize}
\begin{mathpar} 
   \mproj {\{\mathtt{xs1,x1}\}} \cdot \mnil \iret \cdot \mnil {\mathtt{xs1}} \cdot \mshift {\mathtt{lst}} {\mathtt{xs1}} \cdot \mzero {\mathtt{x1}}
   =
\begin{tiny}
\begin{blockarray}{cccc}
  {} &  \mathtt{lst} \idot \ideg 2 & \mathtt{lst} \idot\ideg1 & \iconst   \\
  \begin{block}{c(ccc)}
     \mathtt{x1} \idot \inds{\alpha}  & 0 & 0 & 0 \\
     \mathtt{xs1} \idot \ideg 2 & * & * & * \\
     \mathtt{xs1}  \idot \ideg 1 & * & * & * \\
  \end{block}
  \end{blockarray}
\end{tiny}

   \mproj { \{\mathtt{xs1,x1}\}} \cdot \mproj {\neg \{\mathtt{xs2,x2}\}} \cdot \mproj {\neg \mathtt{tmp}} \cdot \munshift {\mathtt{tmp}} \iret \cdot \mmove \iret {\mathtt{tmp}} \cdot M \cdot \mmove {\mathtt{xs2}} \iarg \cdot \mshift {\mathtt{xs1}} {\mathtt{xs2}} \cdot \mzero {\mathtt{x2}} \cdot \mshift {\mathtt{lst}} {\mathtt{xs1}} \cdot \mzero {\mathtt{x1}}
   =
\begin{tiny}
\begin{blockarray}{cccc}
  {} &  \mathtt{lst} \idot \ideg 2 & \mathtt{lst} \idot\ideg1 & \iconst   \\
  \begin{block}{c(ccc)}
     \mathtt{x1} \idot \inds{\alpha}  & 0 & 0 & 0 \\
     \mathtt{xs1} \idot \ideg 2 & 0 & 0 & 0 \\
     \mathtt{xs1}  \idot \ideg 1 & 0 & 0 & 0 \\
  \end{block}
  \end{blockarray}
\end{tiny}

    \mproj {\{\mathtt{xs2,x2}\}} \cdot \mproj {\neg \mathtt{tmp}} \cdot \munshift {\mathtt{tmp}} \iret \cdot \mmove \iret {\mathtt{tmp}} \cdot M \cdot \mmove {\mathtt{xs2}} \iarg \cdot \mshift {\mathtt{xs1}} {\mathtt{xs2}} \cdot \mzero {\mathtt{x2}} \cdot \mshift {\mathtt{lst}} {\mathtt{xs1}} \cdot \mzero {\mathtt{x1}}
       =
\begin{tiny}
\begin{blockarray}{cccc}
  {} &  \mathtt{lst} \idot \ideg 2 & \mathtt{lst} \idot\ideg1 & \iconst   \\
  \begin{block}{c(ccc)}
     \mathtt{x2} \idot \inds{\alpha}  & 0 & 0 & 0 \\
     \mathtt{xs2} \idot \ideg 2 & 0 & 0 & 0 \\
     \mathtt{xs2}  \idot \ideg 1 & 0 & 0 & 0 \\
  \end{block}
  \end{blockarray}
\end{tiny}

    \mproj {\mathtt{tmp}} \cdot \munshift {\mathtt{tmp}} \iret \cdot \mmove \iret {\mathtt{tmp}} \cdot M \cdot \mmove {\mathtt{xs2}} \iarg \cdot \mshift {\mathtt{xs1}} {\mathtt{xs2}} \cdot \mzero {\mathtt{x2}} \cdot \mshift {\mathtt{lst}} {\mathtt{xs1}} \cdot \mzero {\mathtt{x1}}
     =
\begin{tiny}
\begin{blockarray}{cccc}
  {} &  \mathtt{lst} \idot \ideg 2 & \mathtt{lst} \idot\ideg1 & \iconst   \\
  \begin{block}{c(ccc)}
     \mathtt{tmp} \idot \ideg 2 & 0 & 0 & 0 \\
     \mathtt{tmp}  \idot \ideg 1 & 0 & 0 & 0 \\
  \end{block}
  \end{blockarray}
\end{tiny}
\end{mathpar}
\end{scriptsize}

And one last matrix is generated in $\mathfrak C$ for function 
arguments:
\begin{mathpar}
\mproj {\{ \iarg, \iconst \}} \cdot \mmove {\mathtt{xs2}} \iarg \cdot \mshift {\mathtt{xs1}} {\mathtt{xs2}} \cdot \mzero {\mathtt{x2}} \cdot \mshift {\mathtt{lst}} {\mathtt{xs1}} \cdot \mzero {\mathtt{x1}}
=
\begin{tiny}
\begin{blockarray}{cccc}
  {} &  \mathtt{lst} \idot \ideg 2 & \mathtt{lst} \idot\ideg1 & \iconst   \\
  \begin{block}{c(ccc)}
     \iarg \idot \ideg 2 & 1 & 0 & 0 \\
     \iarg  \idot \ideg 1 & 2 & 1 & 0 \\
     \iconst & 1 & 2 & 1 \\
  \end{block}
  \end{blockarray}
\end{tiny}
\end{mathpar}

The rule \textit{M-Fun} then requires that, for each $C \in \mathfrak C$,
the matrix $C$ is non-negative, i.e.,
\\ %to not cut off inequality
$C \cdot \mzero {\dom{\emptyset}} \geq 0$.
This condition holds for all of the matrices in $\mathfrak C$ (which were shown above),
finishing our justification for typing \texttt{half} as
$\tfun {\tlist \alpha} {\tlist \alpha} M$.
\end{example}

\section{Soundness}
\label{sec:sound}

In this section, we give some intuition behind the obstacles 
to soundness, lay 
out some remaining formalisms necessary to formally state the
soundness theorem, and then go over the theorem.

The key notion of soundness in our type system is that the 
maps we assign to transform annotations do not cause an 
increase in potential---a cost-free type should merely reallocate
potential. However, this property is 
non-trivial
due to the unavoidable presence of negative matrix entries,
like in the primitive map $\munshift x y$.
Such entries can create negative annotations,
posing two problems:

\begin{itemize}
    \item Negative annotations on a variable could mean negative 
    potential on that variable, so that if that variable were to fall
    out of scope potential would be gained.

    \item  Linear maps that merely safely lose potential when applied to positive annotations will symmetrically gain potential when applied to negative annotations.
\end{itemize}

The former of these issues is not
new to amortized-cost-analysis systems
that allow
negative potential, and has been handled previously
by checking that potential is non-negative before
it is dropped
\cite{gueneau2019formal,kahn2021automatic}.
Our system takes
the same approach, but must do so indirectly because our system
does not work directly with annotations.
Instead we use the principle that if both the matrix
$C \in \mathfrak C$ and annotation vector
$\vec p$ are non-negative, then so is $C\cdot \vec p$.
The latter issue is new to our setting of linear maps, but can be handled
similarly.

\subsection{Preliminaries}

\begin{figure} 
\begin{scriptsize}
\begin{mathpar}
\inferrule[V-Fun]{
    \tjudgecf \Gamma {\efun f x e} {\tfun \tau \sigma M} I
    {\emptyset}
    \\\\
    \wf V \Gamma
}{
    \wf {\vclosure V f x e} {\tfun \tau \sigma M}
}

\inferrule[V-Nil]{
}{
    \wf \enil {\tlist \tau}
}

\inferrule[V-Cons]{
    \wf {v_h} \tau 
    \\\\
    \wf {{v_t}} {\tlist \tau}
}{
    \wf {\econs {v_h} {v_t}} {\tlist \tau}
}

\inferrule[V-Bool]{
    v \in \{ \etrue, \efalse\}
}{
    \wf v \tbool
}

\inferrule[V-Context]{
    \forall x \in \dom V.\, \wf {V(x)} {\Gamma(x)}
}{
    \wf V {\Gamma}
}
\end{mathpar}
\end{scriptsize}
\vspace{-.5cm}
\caption{Value well-formedness rules}
\label{fig:val}
\end{figure}

To formulate the soundness statement for our 
type system,
it is first important to formalize what it means for
values to be well-formed and what the operational semantics
of our language is.

\subsubsubsection{\textbf{Well-Formed Values}} We use the judgment $\wf v \tau$ to mean that a value $v$ is well formed at type $\tau$. 
The well-formedness rules are 
given in \Cref{fig:val}. These rules only require
that the value be structurally appropriate for the type, 
except an additional requirement is added for function 
closures. To be well-formed at a type $\tau$, 
a function closure $\vclosure V f x e$
additionally requires that $\efun f x e$ be typable as $\tau$ in some 
type context $\Gamma$ that is well-formed with respect to the captured value environment $V$.

\begin{figure}
\begin{scriptsize}
\begin{mathpar}

\inferrule[E-App]{
    \ejudgenp {V', g\mapsto \vclosure {V'} g y e, y \mapsto v'} e v
}{
    \ejudgenp {V,f\mapsto \vclosure {V'} g y e, x \mapsto v'}
    {\eapp f x} v
}

\inferrule[E-Fun]{
}{
    \ejudgenp V {\efun f x e} {\vclosure V f x e}
}

\inferrule[E-CaseL-Nil]{
    \ejudgenp {V, \ell \mapsto \enil} {e_1} v
}{
    \ejudgenp {V, \ell \mapsto \enil} {\ecasel \ell {e_1} h t {e_2}} v
}

\inferrule[E-CaseL-Cons]{
    \ejudgenp {V, \ell \mapsto \econs {v_h} {v_t}, h \mapsto v_h, t\mapsto v_t} {e_2} v
}{
    \ejudgenp {V, \ell \mapsto \econs {v_h} {v_t}} {\ecasel \ell {e_1} h t {e_2}} v
}

\inferrule[E-Nil]{
}{
    \ejudgenp {V} {\enil} {\enil}
}

\inferrule[E-Cons]{
}{
    \ejudgenp {V,h\mapsto v_h,t\mapsto v_t} {\econs h t} {\econs {v_h} {v_t}}
}

\inferrule[E-Cond-T]{
    \ejudgenp {V, b\mapsto \etrue}  {e_1} v
}{
    \ejudgenp {V, b\mapsto \etrue} {\econd b {e_1} {e_2}} v 
}

\inferrule[E-Cond-F]{
    \ejudgenp {V, b\mapsto \efalse}  {e_2} v
}{
    \ejudgenp {V, b\mapsto \efalse} {\econd b {e_1} {e_2}} v 
}

\inferrule[E-Bool]{
    v \in \{\etrue, \efalse\}
}{
    \ejudgenp {V} v v
}

\inferrule[E-Var]{
}{
    \ejudgenp {V,x\mapsto v} {\evar x} v
}

\inferrule[E-Let]{
    \ejudgenp V {e_1} {v'}
    \\
    \ejudgenp {V,x\mapsto v'} {e_2} v
}{
    \ejudgenp V {\elet {e_1} x {e_2}} v
}
\end{mathpar}
\end{scriptsize}
\caption{Evaluation rules}
\label{fig:op}
\end{figure}

\subsubsubsection{\textbf{Operational Semantics}} We use the judgment $\ejudgenp {V} e v$ to mean that,
under value context $V$, the expression $e$ evaluates 
to the value $v$. The evaluation rules are given in \Cref{fig:op}. 
These are standard strict
evaluation rules, and do not involve any AARA features.

\subsection{Soundness Theorem}

Soundness for our linear-map-based cost-free 
type system is stated as follows:

\begin{theorem}[soundness]
    If 
    \begin{itemize}
        \item $\ejudgenp V e v$ \;\; (the expression $e$ evaluates to $v$ in program context $V$)
        \item $\wf V \Gamma$  \;\; ($V$ is well-formed with respect the type context $\Gamma$)
        \item $\tjudgecf {\Gamma} e \tau {\mathcal S} {\mathfrak C}$  \;\; ($\Gamma$ types 
        $e$ cost-freely as $\tau$ with maps $\mathcal S$ and $\mathfrak C$)
        \item $\mathfrak C \cdot \vec p \geq 0$  \;\;($\mathfrak C$ only
        maps $\vec p$ to non-negative vectors)
        \item $\vec p$ annotates $\Gamma$
    \end{itemize}
    then there is some map $M\in \mathcal S$ such that $\pot {V} {\Gamma} {\vec p} \geq \pot {(V,\iret \mapsto v)} {(\Gamma,\iret : \tau)} {M \cdot \vec p}$, i.e., converting the 
    initial annotation $\vec p$ to an annotation of the
    result by applying $M$ does not gain potential.
\end{theorem}

\begin{proof}
The property is established via
lexicographic induction over the evaluation judgment followed by the typing judgment.
Critically, we make use of the fact that
if $\vec p \geq 0$ and $M \leq N$, then $M \cdot \vec p \leq N \cdot \vec p$. See \Cref{sec:details} for details.
\end{proof}

It is easier to understand why this statement is a good 
formalization of soundness
when specializing the theorem to function applications, which are the real
focus of the cost-free type system:

\begin{corollary}[soundness for function application]
\label{cor:f}
If 
    \begin{itemize}
        \item $\ejudgenp V {\eapp {f} {x}} v$ \;\; (applying function $f$ to argument $x$ evaluates to $v$ in program context $V$)
        \item $\wf V \Gamma$  \;\; ($V$ is well-formed with respect the type context $\Gamma$)
        \item $\Gamma(f) = \tfun \tau \sigma M$ \; and \; $\Gamma(x) = \tau$ \;\; ($\Gamma$ gives $f$ and $x$ these cost-free types)
        \item $\vec p \geq 0$ annotates $x$
    \end{itemize}
    then $\pot {V(x)} \tau {\vec p} \geq \pot v \sigma {M\cdot \vec p}$, i.e.,
    the potential of the argument annotated 
with $\vec p$ is an upper bound on the potential 
of the return annotated with $M \cdot \vec p$.
\end{corollary}

Intuitively, \Cref{cor:f} tells us that a potential annotation 
$\vec p$ on a function argument 
can be safely transformed into an annotation $M\cdot \vec p$
for the function return. In the setting of cost-free typing,
this transformation 
provides a way to reallocate excess amounts of potential
on the argument to the return. And by design, this 
transformation of multiplying by the matrix $M$ involves only linear 
arithmetic, which allows it to integrate into 
the linear programming that AARA already uses.

\section{Automation}
\label{sec:aut}

In this section, we explain how to automate
type checking and inference for our cost-free type system.
In particular, we explain
how type checking is always
efficiently reducible to simple linear arithmetic
and how type inference can usually be reduced to 
linear programming.
%(although it could, in principle, require quadratically-constrained quadratic programming)
We discuss obstacles to further automation in \Cref{sec:lim}.

\subsubsubsection{\textbf{Checking}}
Types can be checked directly in our type system.
The un-annotated base types can be checked via standard 
techniques, and then all that is left is to confirm
the annotations of functions. However,
once each function is annotated with a concrete matrix,
confirming the correctness of those matrices to checking the 3 matrix
inequalities over the sets $\mathcal S$ and $\mathfrak C$ in the rule \textit{M-Fun}, which is mostly
a matter of linear arithmetic. In detail: 

\begin{enumerate}
    \item All the maps of $\mathcal S$ and $\mathfrak C$ must be
    generated through the typing rules' matrix multiplications.

    \item Pointwise inequalities between matrices must be set up 
    and checked according to \textit{M-Fun} (with inequalities over
    havocked elements $*$ from $\mnil x$ being filtered out). 
\end{enumerate}

Thus,
the efficiency of type checking depends only on
how quickly the matrices of the sets $\mathcal S$ and $\mathfrak C$
can be generated,
and how quickly those equalities can be checked.

Let the longest execution path through a function's body be comprised of
$n$ subexpressions, and
let there be up to $k$ lists present in scope 
at a time.
Note the following asymptotic bounds:

\begin{itemize}
    \item $\Theta(n)$ matrix multiplications are used to obtain 
    each $M\in \mathcal S \cup \mathfrak C$
    \item the dimension of all matrices involved is $O(k\cdot d) $
    \item $|\mathcal S| + |\mathfrak C| = O(2^n)$ 
\end{itemize}

These bounds can be combined to provide a bound on
how quickly the matrix inequalities of \textit{M-Fun} can be generated.
Letting $\omega \in [2,3]$ be 
the exponent associate with matrix multiplication, that bound is
$O(2^n \cdot  n \cdot  (k\cdot d)^\omega)$. Then there
are only $O((k \cdot d)^2)$ matrix elements to 
check inequalities over, 
so $O(2^n \cdot  n \cdot  (k\cdot d)^\omega)$ gives an upper bound on
the time complexity of the entire process.

Note that the only
super-polynomial part of this complexity comes from the number of paths through
the function body being $O(2^n)$ in the worst case. In practice,
people do not usually write code in a way that hits this
upper bound. In particular, if no let-bound expression involves branching, then the number of 
paths is $O(n)$, yielding a
time complexity of
$O(n^2 \cdot  (k\cdot d)^\omega)$.

\subsubsubsection{\textbf{Inference}}
The only inherent difference between type inference 
and checking is that the matrices of $\mathcal S$ and 
$\mathfrak C$ now
include unknown variables from the function annotation being
inferred. As a result, the matrix 
inequalities that must hold for the rule \textit{M-Fun} are 
no longer simple arithmetic checks; instead they
require solving systems 
of inequalities over polynomials. 
However, to minimize the number of unknowns involved, one should
also impose the following ordering on type inference:

\begin{enumerate}
    \item Topologically sort functions such that $f \leq g$
    if $f$ is called in the body of $g$.

    \item Infer types in an ascending order.
\end{enumerate}

By inferring types in this order, the type of a function $f$ is 
only inferred after all the functions $f$ depends on have had 
their types inferred. Thus, the polynomial inequalities
only involve unknowns from the matrix annotating the type
of $f$ itself (for non-mutually-recursive $f$).

Such inequalities
can be solved optimally in general
using quadratically-constrained quadratic programming, which
is known to be NP-hard. However, it is common that
such polynomials are actually linear, allowing
an optimal solution to be found in polynomial
time via linear programming. For example, consider the inequalities from the 
example typing \texttt{half} 
in \Cref{sec:trules}---a linear program that appropriately 
maximizes the 
higher-degree entries of $M$ yields the exact matrix
we desired for \texttt{half}. In such a case, 
the effect on time complexity is
that the $O((k\cdot d)^2)$ inequality checks can be replaced with
$O((k\cdot d)^{r})$, where $r \geq \omega$ is the exponent associated
with linear program solving.

One can tell when linear programming is insufficient 
simply by
checking the generated constraint set. 
In such a case, one could either default
to nonlinear constraint solving or go back to the existing cost-free 
approach, depending on the tradeoffs one wishes to make. However,
we expect many cases to amenable to linear programming in 
practice, and we validate this expectation
in \Cref{subsec:realistic}.

To see why linear programming commonly suffices, consider what needs to
happen for nonlinear constraints to arise. Nonlinear terms can only
arise in the course of multiplication between 
two (possibly identical) matrices $M$ and $N$
both containing unknown variables. Such multiplication can only
occur if, along a single path through a function's body, 
two function calls are
made to functions whose types have not yet been inferred.
This circumstance can occur when calling higher-order
functions or when inferring the type
of recursive functions with non-linear recursion schemes.
Much code, including every code example 
given so far in this work, uses a linear recursion scheme,
and thus admits efficient type inference. (See \Cref{sec:lim} for
higher-order function discussion.)

Further, while nonlinear recursion
is \textit{necessary} for nonlinear constraints in a first-order setting, 
it is not \textit{sufficient}.
If $M= M' \oplus I$ and 
$N=I \oplus N'$, then $M\cdot N = M' \oplus N'$, which is linear if 
$M$ and $N$ are. This pattern arises if two recursive
calls are made to independent arguments. Thus, nonlinear constraints
require something like the chaining of function applications to define
the very function being applied. In other words, the typical
example of the problem
case is recursively defining the function \texttt{f} by
using \texttt{f (f x)} somewhere in its body. This pattern
almost never arises in real code.

\section{Experiments}
\label{sec:exp}

In this section, we experiment with our implementation of 
the new
inference algorithm.
First, we compare the algorithm's efficiency to the existing approach on parameterized synthetic code,
which---as expected---demonstrates
the relative exponential improvement in
how the new algorithm scales.
Then we run the
new algorithm on a collection of realistic list functions
to obtain an absolute sense of its performance.
The latter experiment also validates that our new algorithm
usually deals with only linear constraints.
%\twr{Say what language the algorithm is implemented in.}
All experiments were implemented 
in OCaml 4.12.0, 
run on a Mac with a 2.3 GHz Dual-Core Intel Core 
i5 processor, and use the Gurobi version 9.5.1 solver.

\subsection{Efficiency Comparison}

To compare our cost-free type-inference technique with the existing approach \cite{hoffmann2010amortized,HoffmannAH12} empirically,
we measured
the efficiency of
implementations of
both the existing cost-free type-inference algorithm
and our new algorithm from \Cref{sec:aut}.
(A theoretical comparison
of the two algorithms
can be found in
in \Cref{sec:further}.)

\begin{figure}
\begin{lstlisting}[xleftmargin=4.0ex]
    let g2 = 
        let g1 = fun f0 x0 = x0 in
        fun f1 x1 = case x1 of 
            | [] -> [] 
            | h1::t1 -> h1::(g1 (g1 (g1 (f1 t1))))
    in fun f2 x2 = case x2 of 
        | [] -> [] 
        | h2::t2 -> h2::(g2 (g2 (g2 (f2 x2))))
\end{lstlisting}
    \caption{Synthetic code pattern example, $c=3, \ell = 2$}
    \label{fig:synth}
\end{figure}

To measure the efficiency of these algorithms, 
we considered synthetic code exhibiting relevant patterns
that both algorithms are capable of analyzing.
Because these algorithms' runtimes 
are highly dependent on the code
patterns used, we parameterized the code
based on two salient features:
how many calls $c$ are made to a function after the function is 
defined, and the length $\ell$ of the chain of a function's 
dependencies on other functions. Both $c$ and $\ell$ not only measure features salient
for efficiency of the algorithms, but also 
naturally scale with the use of helper functions in real code. 
We also experimented with varying degrees $d$ of potential,
which is another performance-relevant parameter of interest.

The specific code pattern
we generated for our experiments is exemplified in \Cref{fig:synth}, 
where we define
$\ell$ linearly recursive functions, each of which 
calls the previously defined function $c$ times. 
This code is functionally equivalent to the identity 
function on lists.

This experiment was
designed to answer the following questions:

\fbox{\parbox{0.9\linewidth}{
\begin{enumerate}
    \item[Q1.] How well does the existing algorithm scale
    as a function of $c$ the number of calls
    to a helper function, $\ell$ the length of the dependency
    chain, and $d$ the degree?
    \item[Q2.] How well does our new algorithm scale as a function of $c$, $\ell$, and $d$?
    \item[Q3.] Does the overhead of our new algorithm 
    make it less efficient than the existing algorithm
    on simple programs without challenging patterns; if so, to what extent?
\end{enumerate}
}
}

\noindent
Our findings can be summarized as follows:

\fbox{\parbox{0.9\linewidth}{
\begin{enumerate}
    \item[A1.] The existing algorithm scales
    poorly both in terms of the number constraints 
    and time. Its efficiency
    seems to scale exponentially in each of 
    $c$, $\ell$, and $d$. On the largest cases
    in our testing range, this algorithm times
    out after days.
    \item[A2.] In comparison, our new algorithm scales
    well. {Its efficiency does not
    seem to scale exponentially in any of 
    $c$, $\ell$, and $d$.} Our new algorithm always generates
    fewer constraints than the existing 
    algorithm and
    in many cases is multiple
    orders of magnitude more efficient in terms of both time taken and constraints 
    generated. 
    \item[A3.] Despite generating fewer constraints
    than the existing algorithm, our
    new algorithm takes more time to solve those
    constraints for the smallest cases in our
    testing range. However, the difference is never more 
    than a fraction of a second.
\end{enumerate}
}
}

\smallskip
To keep {the experiments} fair, we built our 
own implementation of the existing 
approach.\footnote{We also modify 
this implementation with 
\textit{remainder contexts} \cite{kahn2021automatic,grosen2023automatic},
which aid in sharing. Remainder contexts 
make the implementation simpler 
while increasing analytical power and without significantly
altering the number of constraints.
} 
There is an available implementation of AARA called RaML \cite{hoffmann2012resource},
which uses the existing approach.
However, RaML implements many special features,
including \textit{multivariate} potential,\footnote{
    Multivariate potential is where potential is not 
    expressed merely as sums of basic functions, but also as products.
} which 
would unfairly slow down RaML's analysis. Simultaneously, RaML includes many optimizations to minimize memory
usage and constraint generation that our prototype does not
use, unfairly \textit{advantaging} RaML's analysis. 
Thus, a fresh implementation gives a more direct comparison.

We ran our implementations over the same code 
representations on the same computer with 
the same interface to the same linear-program solver, 
creating a level testing ground. In these experiments, we measured
the number of linear constraints generated, the time taken to 
generate the constraints, and the time taken to solve the
constraints for $1 \leq d \leq 6$ and $0 \leq c,\ell \leq 5$.

\begin{figure}
    \centering
    \includegraphics[scale=0.39]{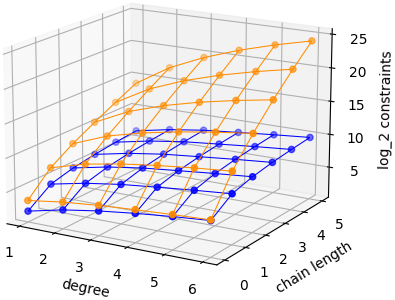}
    \includegraphics[scale=0.39]{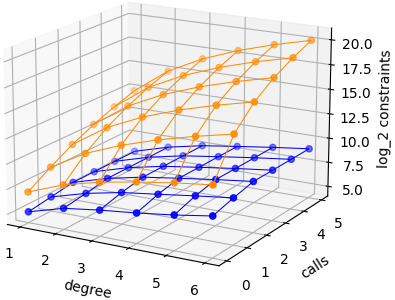}
    \includegraphics[scale=0.39]{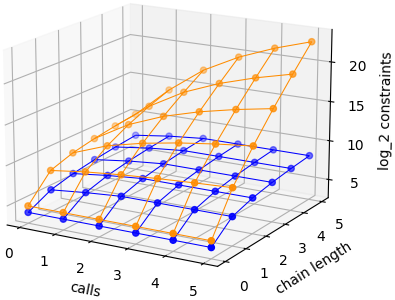}
    \includegraphics[scale=0.39]{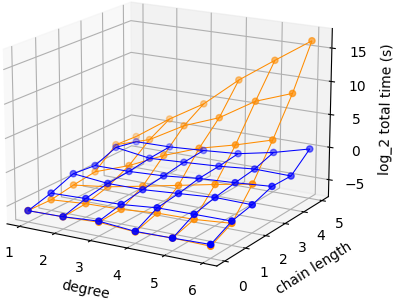}
    \includegraphics[scale=0.39]{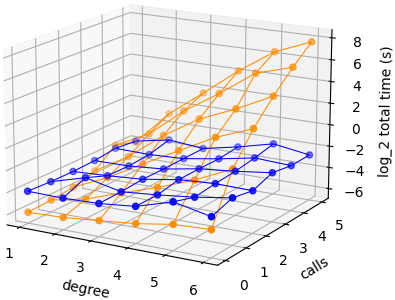}
    \includegraphics[scale=0.39]{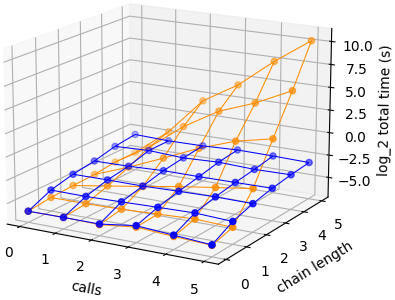}
    \\
    Key: \textcolor{orange}{\textbf{orange}} = existing approach, \textcolor{blue}{\textbf{blue}} = new approach
    \caption{Surface plots of data with one of $d,c,\ell$ fixed at 3.
      (Lower is better.)
    }
    \label{fig:rprdata}
\end{figure}

The complete experimental data can be found in \Cref{table:data1,table:data2,table:data3} in the
supplementary material.
Graphs of some selected features can be found in \Cref{fig:rprdata}.
{
Each surface plot graph shows the scaling 
effects of two of $d$, $c$, and $\ell$ while the third parameter is fixed at 3.
Note that the data points occur at mesh vertices, and that the vertical axis of each graph is on a log scale.
}

Even in this small domain, our experimental results
show that our new algorithm scales significantly better.
While both algorithms had similar runtimes
and generated similar
numbers of constraints
in the smallest cases, 
the exponential growth of the existing algorithm
took over quickly. 
Our implementation of the existing algorithm 
eventually timed out after taking
days on the largest cases, where $d+c+\ell \geq 15$,
while the new algorithm could complete
the analysis of
each of those cases in well under half a second.
One of the largest tests that both of the implementations completed,
with $d=\ell=5$ and $c=4$, saw our algorithm
take 0.13s to create 631 constraints, which were solved
in 0.19s---whereas the existing algorithm took 36s
to generate 27 million constraints, which
were solved in about 150000 seconds (43 hours). 

\subsection {Realistic Examples}
\label{subsec:realistic}

To give an absolute sense of the kind of inference
performance that could be expected on real code, as well
as to validate that nonlinear constraints do not usually 
arise, we ran our inference algorithm on various 
common list functions.
This set of functions
includes some higher-order functions, \texttt{map} and \texttt{filter}, where 
the inference modification outlined in \Cref{sec:lim} was taken.

This experiment was designed to address the following questions
about our inference algorithm:

\fbox{\parbox{0.9\linewidth}{
\begin{enumerate}
    \item[Q1.] How often does the algorithm avoid nonlinear constraints
    on realistic code?
    \item[Q2.] How quickly can the algorithm handle realistic code?
    \item[Q3.] How often does the map inferred by the algorithm
    optimally reallocate potential from input to output?
\end{enumerate}
}
}

\noindent Our findings can be summarized as follows:

\fbox{\parbox{0.9\linewidth}{
\begin{enumerate}
    \item[A1.] None of our code examples generated nonlinear constraints.
    \item[A2.] Even while working with a high degree of potential 
    (10), most code could be analyzed in under a second.
    \item[A3.] Non-optimal reallocation usually occurs in code 
    that manipulates
    multiple lists.
\end{enumerate}
}
}

The data from our
experiments can be found in \Cref{table:listfuns}. 
In particular, the table data includes
the time spent to generate constraints from the 
source code, the time spent to solve
those constraints, the total time, the total number of 
constraints generated, whether 
inference successfully yields a type, and whether the inferred 
linear map optimally reallocates potential.

\begin{table}[]
\begin{footnotesize}
\begin{tabular}{r||ccc|cc|cc}
    function & constr secs & solve secs & total secs & constr count
    & linear & infer success & optimal reallocation
    \\ \hline
    cons & 0.011 & 0.032 & 0.043 & 122 & \checkmark & \checkmark & \checkmark
    \\ 
    uncons & 0.019 & 0.029 & 0.048 & 122 & \checkmark & \checkmark & \checkmark 
    \\ 
    map & 0.036 & 0.068 & 0.104 & 158 & \checkmark & \checkmark & \checkmark 
    \\ 
    filter & 0.078 & 0.049 & 0.127 & 278& \checkmark & \checkmark & \checkmark 
    \\ 
    zip & 0.338 & 0.082 & 0.420 & 359& \checkmark & \checkmark & \checkmark 
    \\
    unzip &  0.266 & 0.081 & 0.347 & 255& \checkmark & \checkmark & \checkmark
    \\
    insert & 0.096 & 0.034 & 0.130 & 278 & \checkmark & \checkmark & \checkmark
    \\
    remove & 0.050 & 0.040 & 0.090 & 267 & \checkmark & \checkmark & \checkmark
    \\
    insertion sort & 0.117 & 0.070 & 0.186 & 544& \checkmark & \checkmark & \checkmark 
    \\
    split & 0.047 & 0.097 & 0.144 & 266 & \checkmark & \checkmark & \checkmark
    \\
    merge & 0.750 & 0.091 & 0.841 & 1009 & \checkmark & \checkmark & X
    \\
    merge sort & 2.307 & 0.329 & 2.636 & 1703 & \checkmark & \checkmark & X
\end{tabular}
\caption{The inference algorithm's performance on list functions 
at max polynomial degree 10.\vspace{-0.5cm}}
    \label{table:listfuns}
\end{footnotesize}
\end{table}

The results of our experiment show some benefits and drawbacks of
our approach to inference.
First, our approach is quite efficient,
even when working with a high degree of potential.
Second, while quadratic constraints could theoretically arise out of our approach, this experiment supports
the claim
that the mitigating factors discussed in
\Cref{sec:aut} really do keep constraints linear in practice. Even in
the case of \textit{merge sort}, which does not use
a linear recursion scheme, only linear constraints arose. 
Finally, however, the analysis of 
\textit{merge} (and thus also \textit{merge sort}) 
shows that working with
multiple lists can be a source of non-optimal reallocation of potential,
which results in looser cost bounds.
Specifically, the inferred cost-free matrix in 
these cases is only able to
reallocate constant potential, losing all higher-degree potential.
We dive deeper into why this reallocation fails in \Cref{sec:further}. Note, however,
that the maps for \textit{zip} and \textit{unzip} both optimally 
reallocate potential despite
the fact that they each manipulate
multiple lists.

\section{Further Comparison}
\label{sec:further}
To further compare our cost-free inference 
algorithm with the 
existing algorithm, this section examines the
theoretical behaviour of both
in terms of asymptotic scaling and 
cost-bound tightness. 

\subsubsubsection{\textbf{Asymptotics}}
Here we give bounds on the number of linear constraints generated
in terms of $d$, $c$, and $\ell$ for both
the existing inference algorithm and our new algorithm. 
These bounds can be compared to the empirical data obtained 
in \Cref{sec:exp}. Theoretical bounds on runtimes
can be obtained from these bounds on 
constraint counts by composing the bounds with
the complexity of linear program
solving. As the choice of linear program solver is not
essential to our type system, we do not further consider
runtime in this section.

To analyze the existing algorithm, we note the
following recurrence for $T(d,c,\ell)$, the worst-case number of constraints that the algorithm generates: 
\begin{align} 
    T(d,c,0) & = O(d)\label{base l}
    \\
    T(1,c,\ell) &= c \cdot T(1,c,\ell-1) + O(c)\label{base d}
    \\
    T(d,c,\ell) &= c \cdot T(d,c,\ell-1) + T(d-1,c,\ell) + O(cd) \label{recur}
\end{align}

This recurrence arises because: 
\begin{enumerate} 
    \item If $\ell=0$ 
    the code generated is simply the non-recursive identity function. The number of constraints 
    generated is therefore proportional 
    the size of the annotation vector $O(d)$.

    \item If $d=1$ then the algorithm needs to (i) find specialized types for each of the $c$ calls to 
    helper functions
    (with one
    less $\ell$), as well as (ii) type its own body.
    In typing the body, in the worst case, the algorithm
    generates constraints proportional to the body length $O(c)=O(cd)$.

    \item In larger cases, the algorithm needs to do all that
    was discussed
    above,
    in addition to retyping the function
    body with one less degree of potential.
\end{enumerate}
While we know no closed form for this recurrence,
we can count the number of occurrences of $T(d-i,c,\ell-j)$
as $\binom {i+j} i \cdot c^j \leq \binom {i+j} i \cdot c^\ell$, and overapproximate the constraints
generated per occurrence as $O(cd)$ (assuming $c \geq 1$). Then 
$T(d,c,\ell) \leq dc^{\ell+1}\sum_{i=0}^{d-1}\sum_{j=0}^{\ell}\binom {i+j} i = dc^{\ell+1} (\binom {d+\ell+1} d -1) = O(dc^{\ell+1}\binom {d+\ell+1} d )$ gives
an upper bound on the number of constraints generated.
While this is only an upper bound, its
poor scaling in all three considered parameters $c,d,\ell$
is qualitatively consistent with
our experimental findings from \Cref{sec:exp}.

Analyzing our new algorithm's constraint generation on the
same patterns is 
more straightforward.
Because each function is analyzed only once,
the worst-case constraints generated are proportional 
to the number of functions to analyze $O(\ell)$. In each
function, the number of non-negativity constraints
is at most proportional to the number of calls $O(c)$
and the size of the symbolic matrix $O(d^2)$.
{
Thus, in the worst-case, the number of constraints generated is $O(d^2c\ell)$.
}

We also consider the behavior as $c$ and $\ell$ get 
\textit{small} to compare performance on less extreme
cases. We do so by analyzing the case where $c=\ell=0$,
which corresponds to
a linearly recursive function calling no
helper functions. In such a case, the two algorithms
both generate $\Theta(d^2)$ constraints. 
For the existing algorithm, the constraint count is $\Theta(d^2)$ because the constraints are generated in proportion to the size of the annotation vector being constrained, which is $\Theta(d)$ in a single typing pass, and $\Theta(d)$ passes are made.
For the new algorithm, the constraint count is $\Theta(d^2)$ because the constraints generated are proportional to the size of the matrix being constrained, which is $\Theta(d^2)$.

\subsubsubsection{{\textbf{Bound Tightness}}}
Ideally, a cost-free function type would tightly express
the cost of 0 between its argument and return types,
because no potential should be spent. 
However, AARA cannot always perfectly reallocate potential 
across different data structures, so potential might
be lost instead. The less potential 
that is lost, the tighter the cost bound that AARA can 
infer.

The existing approach for inferring cost-free types 
(when it succeeds) yields the tightest bounds that AARA can get. This optimality 
occurs due to brute force: the existing
approach infers a new cost-free type 
wherever such a type could better
specialize to the present potential. Thus, as long as the approach's key
assumption is met and
inference succeeds, the type found allocates potential as 
losslessly as AARA types can 
represent. 

In contrast, our new approach may yield looser cost bounds. While 
our approach does well for code
that manipulates only one list, code that
manipulates multiple lists may introduce more loss (as in \Cref{table:listfuns}).
This loss occurs because it is easy to use multiple 
lists to make certain type
annotations dependent upon the minimum
of the annotations across the lists, but 
minimum is not a linear function. The following projection 
function \texttt{proj} provides an
example of such loss because that function would be best typed as
$\tfunshow {\potp {\tprod \tbool {\tprod {\tlistshow a \tau} {\tlistshow b \tau}}} 0} {\potp {\tlistshow {\mathit{min}(a,b)} \tau} 0}$.
Our new approach can only infer the trivial type
$\tfunshow {\potp {\tprod \tbool {\tprod {\tlistshow a \tau} {\tlistshow b \tau}}} 0} {\potp {\tlistshow {0} \tau} 0}$, where all potential is lost,
because the constant zero function is the best linear function
under-approximating the minimum. We hope such nonlinear functions can be addressed in future work.

\begin{lstlisting}[numbers=none]
    fun proj (b, lst1, lst2) = if b then lst1 else lst2
\end{lstlisting}

\section{Limitations and Discussion}
\label{sec:lim}

To keep our types inferable via linear programming, we have made certain 
limiting design decisions. Some of these have already been
discussed in other sections, including the minor limitations
imposed by not having potential on nested lists and by
tracking sets of matrices $\mathcal S$ and $\mathfrak C$ rather
than single matrices. 
However, others design decisions have not
yet been fully discussed.
Such decisions include 
not supporting arbitrary ways of splitting potential
between multiple uses of data structures ("full sharing")
and not supporting higher-order functions.
Here we will discuss the extent of these remaining limitations,
as well as some possible routes for generalizing
beyond them.

\subsubsubsection{\textbf{Higher-Order Functions}} 
During type inference, higher-order function arguments have both 
an unknown type annotation and an unknown body. As a result, it
is unclear what annotation is appropriate to infer. However,
this problem is not new in the setting of cost analysis. 
The cost of
running a higher-order function is usually dependent on
the cost of the function it takes as input, and therefore
would require a symbolic cost bound parameterized by
the input's cost behaviour. 
AARA cannot currently express such a bound.

To deal with this problem, previous AARA literature 
treats
higher-order functions with unknown
arguments as if their argument functions have zero cost.
This treatment has the effect of only reporting
the cost behaviour of the higher-order function itself.
Then when given a concretely-typed argument function,
AARA will macro-expand and retype the higher-order function to take
the known cost behaviour into account. Our technique
can handle higher-order functions in the same way
without raising any additional conceptual issues.

\subsubsubsection{\textbf{Full Sharing}} 
Potential on a variable must be 
split between all its uses. Thus, 
if standard AARA would annotate $x$ with $\vec a$, then 
$\vec b\geq 0$ could go to one use and 
$\vec c\geq 0$ to another, where $\vec a = \vec b + \vec c$. 

While such sharing constraints are linear, the corresponding
\textit{map} acting on $\vec a$ 
introduces nonlinearity. 
Such a sharing map looks like $p \cdot \mmove x y + (1-p) \cdot \mmove x z$, where
$y$ and $z$ stand in for two
uses of $x$ and $p \in [0,1]$. This map 
involves the unknown $p$, which means that composing
with other matrices
involving unknowns (like those for recursive calls) can result in nonlinear terms.

To avoid this problem, our type system codifies the 
essentially-arbitrary choice that all potential is allocated to a
variable's first use. This choice is sufficient
for typing many functions of interest. In addition, our 
implementation augments the type system with \textit{remainder contexts}
\cite{kahn2020exponential,grosen2023automatic},
which can recover
unused potential in certain circumstances, regaining
some of the benefits of full sharing.

\subsubsubsection{\textbf{Nonlinear Constraints}} 
Given how often nonlinear constraints
have come up,
it is natural to ask whether nonlinear
constraint solvers are efficient enough to reasonably
automate our type inference. Indeed, 
we explored using
quadratically-constrained quadratic programs (QCQPs) to
automate type inference. QCQPs are
like a linear program extended to allow quadratic terms,
and this power is sufficient to handle general
polynomial constraints as well.
However, unlike linear programs,
QCQPs are NP-hard to solve outside of special cases,
most notably semidefinite programs \cite{vandenberghe1996semidefinite}.
Unfortunately,
our constraints are not semidefinite, so the efficiency of
QCQP solving is a real issue.
When we experimented with a general purpose QCQP solver,
we had little success.
On simple cases,
a prototype implementation ran {for several days} without finding a solution until it timed out.
Nonetheless, \Cref{sec:exp} suggests this
level of efficiency might be competitive with the existing
approach, so this avenue
is something to consider in the future,
especially as QCQP solvers 
improve. Until then, maximizing the use of
linear constraints is the key to efficient
inference.

A second way to handle nonlinear constraints
is an iterative approach like Kleene iteration \cite{cousot1977abstract} or Newton's method \cite{esparza2010newtonian}.
Conceptually, such an algorithm obtains a solution for a set of nonlinear constraints by creating a set of recursive equations to be solved for a fixed point, fixing an initial guess, and iterating.
Each iteration should yield
a ``better'' approximation to the solution, as defined by some partial order. 
Because each approximate solution is concrete, this
process can be set up in our setting so that each iteration
only requires solving linear programs. 
However, our attempts to devise such an iterative strategy were unsuccessful because 
our system must support matrices with negative entries, like
the unshift map $\munshift x y$.
When composed, such maps transform positive entries of matrices into negative ones, and vice versa.
As a result, any partial ordering on approximate solutions
that respects iteration does not properly respect pointwise ordering.
Iteration uses the former order, but our linear programs use the latter, so they cannot be combined.

\section{Related Work}
\label{sec:relatedwork}

\subsubsubsection{\textbf{Amortized Cost Analyses}} 
There are many other program analyses that 
deal with amortized cost bounds.
These include program logics augmented with cost credits 
\cite{atkey2010amortised,gueneau2018fistful},
other type systems \cite{rajani2021unifying},
recurrence-relation solving \cite{cutler2020denotational},
term rewriting \cite{hofmann2014amortised,moser2020automated},
and more 
\cite{batz2023calculus,lu2021selectively,nipkow2019amortized,sinn2014simple,flores2016upper,fiedor2018shapes}.
However, we do not know of any non-AARA-based systems
that currently exploit the ideas behind 
cost-free types.\footnote{
For a converse example, the work of Moser and Schneckenreither
\cite{moser2020automated} is based on AARA and
uses the concept of a cost-free typing.
Their paper
also makes the observation
that cost-freedom can make constraint sets blow up 
in size.
}
The main purpose of cost-freedom is to describe soundly how 
excess potential should be allocated when composing
amortized-cost analyses. However, it is always just as sound
to simply drop excess potential or reanalyze
the program, so less automated or 
less mature systems may not have found
cost-freedom necessary.
Furthermore, the stickier problem solved by
cost-freedom, resource-polymorphic recursion, only occurs
when a variety of features come together: some form
of resource credits (like potential), nonlinear cost bounds,
and non-tail recursion. Notably, loop programs are
always tail recursive, thus they never encounter resource-polymorphic recursion.

\subsubsubsection{\textbf{Automated Cost Analyses}}
Our work is about AARA,
but there are many other approaches to automated
cost analysis. Some such systems include 
COSTA \cite{albert2008costa}, SPEED \cite{gulwani2009speed},
LOOPUS \cite{sinn2014simple}, CoFloCo \cite{flores2016upper}, 
and CHORA \cite{breck2020templates}.
These systems all operate by extracting some sort of
constraints and/or recurrences and passing them
to a solver, similarly to AARA. 
More complex problems are harder to solve,
so these systems all face an inherent trade-off
among
their levels of automation, speed, expressivity,
and completeness. Note also that our cost-free type system 
is related to recurrence-solving in that both 
mechanisms
infer cost-related functions from functional equations.

\subsubsubsection{\textbf{Linear Fixed Points}}
Our work uses linear programming to solve recursive
constraints over matrices, but there are many 
other approaches in program analysis 
that solve similar linear, fixpoint-like problems. 
Such problems have been studied since at least the 70s
with the introduction of the polyhedral
domain for abstract interpretation, which finds linear
invariants of programs \cite{cousot1978automatic}.
Abstract-interpretation work since then has focused more specifically on  inferring linear or affine transformations 
\cite{muller2007analysis,sharma2019new}, although
nothing that we know of is quite sufficient for
our purposes.
Some other relevant work includes that of
de Oliveira et al., which represents polynomials with a linear 
basis much like AARA, and uses eigenvectors of linear maps
to find loop invariants \cite{de2017synthesizing}. 
Unfortunately,
extending this idea to our domain where inequalities
matter quickly runs into the positivity problem 
(which is not known to be decidable \cite{ouaknine2012decision}).
As part of our experiments with iterative methods, discussed in \Cref{sec:lim}, we also considered the work of
Reps et al. \cite{DBLP:journals/toplas/RepsTP17}, which uses a tensor-product operation to transform linear context-free equations into regular-language equations, allowing each Newton iteration of Newtonian program analysis \cite{esparza2010newtonian} to be solved efficiently.
In our setting, however, tensors do not change the constraints generated.

\iffalse
\subsubsubsection{\textbf{Linear Maps and Linear Languages}}
We used linear maps to capture some of the
linear behaviour of the AARA type system,
but we are not the first to connect these different
flavors of linearity. For example, Valiron and 
Zdancewic \cite{valiron2014finite} gave a denotational 
semantics for a simple functional language using finite 
vector spaces. Unlike our goal of cost analysis, their goal 
was to model linear logic,
and thus their language excludes, e.g., lists.
For a more closely related example, we can
look within the AARA literature and find that there was 
already some acknowledgement that AARA's
potential manipulations could be described 
using linear operators. 
Grosen et al. \cite{grosen2023automatic}
make use of the semimodule abstraction (which is 
similar to a vector space) to more succinctly describe
AARA, and leverage the linearity of operators to simplify 
proof obligations. However, cost-free types in that 
work are still treated as infinite families of types.
\fi

\section{Conclusion}
\label{sec:conc}
Our new type system leverages linear maps to 
derive cost-free types, which are critical
for analyzing composition and non-tail recursion
in AARA. 
%We prove that this type system is sound
%and explain how to automate its checking and
%xinference.
By carefully designing our type system, we can often reduce
type inference to linear programming.
This new cost-free inference avoids
the infinite regress of resource-polymorphic recursion 
by capturing the information of an infinite 
space of
cost-free types with a single matrix. As a result, our
technique scales exponentially better than existing methods
wherever our technique can be applied. This technique is also 
able to work with exponential cost bounds, whereas 
existing AARA methods were confined to polynomials.
Future work could explore using nonlinear constraint
solving to generalize
our approach, as well as determine how best to handle
nonlinear potential transformations like minima.

%\begin{acks}
% The authors would like to acknowledge Kevin Pratt for directing
% attention to Skolem's problem.
%\end{acks}

\bibliographystyle{ACM-Reference-Format}
\bibliography{sources}

\appendix

\section{Soundness Theorem} \label{sec:details}

\begin{theorem}
    If 
    \begin{itemize}
        \item $\ejudgenp V e v$ \;\; (the expression $e$ evaluates to $v$ in program context $V$)
        \item $\wf V \Gamma$  \;\; ($V$ is well-formed with respect the type context $\Gamma$)
        \item $\tjudgecf {\Gamma} e \tau {\mathcal S} {\mathfrak C}$  \;\; ($\Gamma$ types 
        $e$ cost-freely as $\tau$ with maps $\mathcal S$ and $\mathfrak C$)
        \item $\mathfrak C \cdot \vec p \geq 0$  \;\;($\mathfrak C$ only
        maps $\vec p$ to non-negative vectors)
        \item $\vec p$ annotates $\Gamma$
    \end{itemize}
    then there is some map $M\in \mathcal S$ such that $\pot {V} {\Gamma} {\vec p} \geq \pot {(V,\iret \mapsto v)} {(\Gamma,\iret : \tau)} {M \cdot \vec p}$, i.e., converting the 
    initial annotation $\vec p$ to an annotation of the
    result by applying $M$ does not gain potential
\end{theorem}

\begin{proof}

We proceed by lexicographic induction over the derivation of 
$\ejudgenp V e v$, followed by the derivation of 
$\tjudgecf {\Gamma} e \tau {\mathcal S} {\mathfrak C}$. 

\subsubsubsection{\textbf{application}}

Suppose that \textit{E-App} was the last rule 
applied to derive
$\ejudgenp {V, f \mapsto \vclosure {V'} g y e, x\mapsto v'} {\eapp f x} v$.

\[
\inferrule[E-App]{
    \ejudgenp {V', g\mapsto \vclosure {V'} g y e, y \mapsto v'} e v
}{
    \ejudgenp {V,f\mapsto \vclosure {V'} g y e, x \mapsto v'}
    {\eapp f x} v
}
\]

Then we know the rule's premise holds by inversion. 

At the same time, only one typing rule, \textit{M-App}, types the 
application expression, so \textit{M-App} must have been the last rule 
applied.

\[\inferrule[M-App]{
}{
    \tjudgecf {\Gamma, f: \tfun \tau \sigma M, x: \tau} {\eapp f x} \sigma { \{ M\cdot \mmove x \iarg \} } {\{ \mproj { \{ x,\iconst \} } \}}
}\]

This tells us that $\mathfrak C = \{ \mproj { \{ x, \iconst \} } \}$. As a result, the condition that $\mathfrak C \cdot \vec p \geq 0$ 
just amounts to $\mproj { \{ x, \iconst \} } \cdot \vec p \geq 0$.

Because we know the whole context is well-formed by assumption, 
we know in particular
that
$\wf {\vclosure {V'} g y e} {\tfun \tau \sigma  M} $, and so
the well-formedness rule for function values must
apply.

\[
\inferrule[V-Fun]{
    \tjudgecf {\Gamma'} {\efun g y e} {\tfun \tau \sigma M} I
    {\emptyset}
    \\
    \wf {V'} {\Gamma'}
}{
    \wf {\vclosure {V'} g y e} {\tfun \tau \sigma M}
}
\]

Thus by inversion, there must also exist a typing 
derivation for the function $f$. Only one rule, \textit{M-Fun}, 
derives the function type, so \textit{M-Fun} must have 
been the last rule applied for the derivation of 
that typing judgment, and we know its premises hold by inversion.

\[
\inferrule[M-Fun]{
    \tjudgecf {\Gamma', g: \tfun \tau \sigma M, y: \tau} e \sigma {\mathcal S'} {\mathfrak C'}
    \\
     \mproj  {\{y, \iret,\iconst \} } \cdot \mathcal S' \cdot \mzero {\dom{\Gamma'}} 
     \geq \mproj {\{ y, \iret,\iconst \} } \cdot M \cdot \mmove y \iarg 
    \\
    \mproj  {\neg \{y, \iret,\iconst \} } \cdot \mathcal S' \cdot \mzero {\dom{\Gamma'}} \geq 0
    \\
    \mathfrak C' \cdot \mzero {\dom{\Gamma'}} \geq 0
}{
    \tjudgecf {\Gamma'} {\efun g y e} {\tfun \tau \sigma M}
    {\{ \mi \}} {\emptyset }
}
\]

Now note the following chain of implications which 
we will make use of later:
\begin{align*}
    \mproj { \{ x, \iconst \} } \cdot \vec p \geq 0 & \implies \mmove x y \cdot \mproj { \{ x, \iconst \} } \cdot \vec p \geq 0 & \textrm{relabeling}
    \\
    & \implies \mathfrak C'\cdot \mzero {\dom {\Gamma'}} \cdot \mmove x y \cdot \mproj { \{ x, \iconst \} } \cdot \vec p \geq 0 &  \mathfrak C'\cdot \mzero {\dom {\Gamma'}} \geq 0
\end{align*}

Letting $\mmove x y \cdot \mproj { \{ x, \iconst \} } \cdot \vec p = \vec q$, we can now collect the following statements:
\begin{itemize}
    \item $\ejudgenp {V',  g \mapsto \vclosure {V'} g y e, y\mapsto v'} e v$ from a premise
    
    \item $\wf {\{V',  g \mapsto \vclosure {V'} g y e, y\mapsto v'\}} {\{\Gamma', g: \tfun \tau \sigma M, y: \tau\}}$
    by applying \textit{V-Context}
    
    \item $\tjudgecf {\Gamma', g: \tfun \tau \sigma M, y: \tau} e \sigma 
    {\mathcal S'} {\mathfrak C'}$ from a premise
    
    \item $\mathfrak C'\cdot \mzero {\dom {\Gamma'}} \cdot \vec q \geq 0$ because $\mproj x \cdot p \geq 0$
    by assumption
    
    \item $\mzero {\dom {\Gamma'}} \cdot \vec q$ annotates $\Gamma', g: \tfun \tau \sigma M, y: \tau$
\end{itemize}

We may therefore apply our inductive hypothesis
to complete this case as follows, 
letting $N$ be the map guaranteed in $\mathcal S'$, and
recalling $\vec q \geq 0$:

\scalebox{0.6}{
\parbox{\linewidth}{
\begin{align*} 
& \pot {\{V',  g \mapsto \vclosure {V'} g y e, y\mapsto v'\}} {\{\Gamma', g: \tfun \tau \sigma M, y: \tau\}} { \mzero {\dom {\Gamma'}} \cdot \vec q} 
\\
&\geq \pot {\{V',  g \mapsto \vclosure {V'} g y e, y\mapsto v', \iret\mapsto v\}} {\{\Gamma', g: \tfun \tau \sigma M, y: \tau,\iret : \sigma\}} {N \cdot\mzero {\dom {\Gamma'}} \cdot \vec q}
& IH
\\
\implies 
& \pot { \{ y\mapsto v'\}} { \{y: \tau\}} {\vec q}
\\
&\geq \pot {\{V', y\mapsto v', \iret\mapsto v\}} {\{\Gamma', y: \tau,\iret : \sigma\}} { N \cdot\mzero {\dom {\Gamma'}} \cdot \vec q}
& \textrm{0 potential assigned}
\\
&\geq \pot {\{ y\mapsto v', \iret\mapsto v\}} {\{ y: \tau,\iret : \sigma\}} {\mproj  { \{y, \iret,\iconst \} } \cdot N \cdot\mzero {\dom {\Gamma'}} \cdot \vec q}
&  \mproj  {\neg \{y, \iret,\iconst \} } \cdot N \cdot \mzero {\dom{\Gamma'}} \geq 0
\\
&\geq \pot {\{ y\mapsto v', \iret\mapsto v\}} {\{ y: \tau,\iret : \sigma\}} {\mproj {\{y, \iret, \iconst\}} \cdot M \cdot \mmove y \iarg \cdot \vec q}
& \mproj  {\{ y, \iret,\iconst \} } \cdot N \cdot \mzero {\dom{\Gamma'}}
\geq \mproj {\{y, \iret, \iconst\}} \cdot M \cdot \mmove y \iarg
\\
&=\pot {\{ y\mapsto v', \iret\mapsto v\}} {\{ y: \tau,\iret : \sigma\}} { M \cdot \mmove y \iarg \cdot \vec q}
& \textrm{indices not present}
\\
\implies &
 \pot { \{ y\mapsto v'\}} { \{y: \tau\}} { \mmove x y \cdot \mproj x \cdot \vec p }
 \\
 & \geq \pot {\{ y\mapsto v', \iret\mapsto v\}} {\{ y: \tau,\iret : \sigma\}} {M \cdot \mmove y \iarg \cdot \mmove x y \cdot \mproj x \cdot \vec p }
 & \mmove x y \cdot \mproj { \{ x, \iconst \} } \cdot \vec p = \vec q
 \\
\implies &
 \pot { \{ x\mapsto v'\}} { \{x: \tau\}} { \mmove x x \cdot \mproj { \{ x, \iconst \} } \cdot \vec p }
 \\
 & \geq \pot {\{ x\mapsto v', \iret\mapsto v\}} {\{ x: \tau,\iret : \sigma\}} {M \cdot \mmove x \iarg \cdot \mmove x x \cdot \mproj { \{ x, \iconst \} } \cdot \vec p }
 & \textrm{relabeling}
 \\
 \implies &
 \pot { \{ x\mapsto v'\}} { \{x: \tau\}} { \mproj { \{ x, \iconst \} } \cdot \vec p }
 \\
 & \geq \pot {\{ x\mapsto v', \iret\mapsto v\}} {\{ x: \tau,\iret : \sigma\}} {M \cdot \mmove x \iarg\cdot \mproj { \{ x, \iconst \} } \cdot \vec p }
 & \textrm{simplification}
 \\
 \implies &
 \pot { \{ x\mapsto v'\}} { \{x: \tau\}} { \vec p }
 \\
 & \geq \pot {\{ x\mapsto v', \iret\mapsto v\}} {\{ x: \tau,\iret : \sigma\}} {M \cdot \mmove x \iarg\cdot \vec p }
 & \textrm{indices not present}
 \\
\implies & \pot {V, f\mapsto \vclosure {V'} g y e, x\mapsto v'} {\Gamma, f: \tfun \tau \sigma M, x: \tau} {\vec p} 
\\
&\geq \pot {V, f\mapsto \vclosure {V'} g y e, x\mapsto v', \iret \mapsto v} {\Gamma, f: \tfun \tau \sigma M, x: \tau, \iret : \tau} {M \cdot \mmove x \iarg \cdot \vec p}
& \textrm{adding like terms}
\end{align*}
}
}

Because $\mathcal S$ contains only $M \cdot \mmove x \iarg$, this case is
complete.

\subsubsubsection{\textbf{case list empty}}

Suppose that \textit{E-CaseL-Nil} was the last evaluation rule 
applied. Then we know the rule's premise must hold by inversion. 

\[
\inferrule[E-CaseL-Nil]{
    \ejudgenp {V, \ell \mapsto \enil} {e_1} v
}{
    \ejudgenp {V, \ell \mapsto \enil} {\ecasel \ell {e_1} h t {e_2}} v
}
\]

Only one typing rule, \textit{M-CaseL}, types the list 
casing expression, so \textit{M-CaseL} must have been the 
last rule applied for the derivation of
the typing judgment. Thus by inversion,
we know its premises must hold as well. In particular,
the left side premise holds.

\begin{tiny}
\[
\inferrule[M-CaseL]{
    \tjudgecf {\Gamma, \ell : \tlist \tau } {e_1} \sigma {\mathcal S} {\mathfrak C }
    \\
    \tjudgecf {\Gamma,  \ell : \tlist \tau, h:\tau, t : \tlist \tau } {e_2} \sigma {\mathcal T}
    {\mathfrak D}
}{
    \tjudgecf {\Gamma, \ell : \tlist \tau } {\ecasel \ell {e_1} h t {e_2}}
    \sigma
    {
        \mathcal S \cdot \mnil \ell 
        \cup \mproj {\neg \{h,t\} } \cdot \mathcal T \cdot \mshift \ell t \cdot \mzero h
    }
    {
    \mathfrak C \cdot \mnil \ell
    \cup 
    (\mproj {\{ h, t\}} \cdot \mathcal T \cup \mathfrak D) \cdot \mshift \ell t\cdot \mzero h
    }
}
\]
\end{tiny}

Now because $(\mathfrak C \cdot \mnil \ell
    \cup 
    (\mproj {\{ h, t\}} \cdot \mathcal T \cup \mathfrak D) \cdot \mshift \ell t\cdot \mzero h) \cdot \vec p \geq 0$ by assumption,
    we know in particular that 
    $\mathfrak C \cdot \mnil \ell \cdot \vec p \geq 0$. This 
    tells us that $\mnil \ell \cdot \vec p$ is a vector we may use
    with our inductive hypothesis.
    
    Applying the inductive hypothesis allows us to find 
an $M \in {\mathcal S}$ where the following holds:

\begin{align*}
\pot {(\iret \mapsto v)} {(\iret :\sigma)} 
{M \cdot \mnil \ell \cdot \vec p} \leq 
& \pot {V, \ell \mapsto \enil} 
{\Gamma, \ell : \tlist \tau } 
{\mnil \ell \vec p} & IH
\\
=&\pot {V, \ell \mapsto \enil} {\Gamma, \ell : \tlist \tau } 
{\vec p} & \textrm{empty list 0 potential}
\end{align*}

Thus $M  \cdot \mnil \ell \in \mathcal S \cdot \mnil \ell$ is a sufficient matrix to witness this case.

\subsubsubsection{\textbf{case list cons}}

Suppose that \textit{E-CaseL-Cons} was the last evaluation rule 
applied. Then we know the rule's premise must hold by inversion.

\[
\inferrule[E-CaseL-Cons]{
    \ejudgenp {V, \ell \mapsto \econs {v_h} {v_t}, h \mapsto v_h, t\mapsto v_t} {e_2} v
}{
    \ejudgenp {V, \ell \mapsto \econs {v_h} {v_t}} {\ecasel \ell {e_1} h t {e_2}} v
}
\]

Only one typing rule, \textit{M-CaseL}, types the list 
casing expression, so \textit{M-CaseL} must have been the 
last typing rule applied. Thus by inversion, we also know 
this rule's premises. In particular, we know the right side premise holds.

\begin{tiny}
\[
\inferrule[M-CaseL]{
    \tjudgecf {\Gamma, \ell : \tlist \tau } {e_1} \sigma {\mathcal S} {\mathfrak C }
    \\
    \tjudgecf {\Gamma,  \ell : \tlist \tau, h:\tau, t : \tlist \tau } {e_2} \sigma {\mathcal T}
    {\mathfrak D}
}{
    \tjudgecf {\Gamma, \ell : \tlist \tau } {\ecasel \ell {e_1} h t {e_2}}
    \sigma
    {
        \mathcal S \cdot \mnil \ell 
        \cup \mproj {\neg \{h,t\} } \cdot \mathcal T \cdot \mshift \ell t \cdot \mzero h
    }
    {
    \mathfrak C \cdot \mnil \ell
    \cup 
    (\mproj {\{ h, t\}} \cdot \mathcal T \cup \mathfrak D) \cdot \mshift \ell t\cdot \mzero h
    }
}
\]
\end{tiny}

Now because $(\mathfrak C \cdot \mnil \ell
    \cup 
    (\mproj {\{ h, t\}} \cdot \mathcal T \cup \mathfrak D) \cdot \mshift \ell t\cdot \mzero h) \cdot \vec p \geq 0$ by assumption,
    we know in particular that 
    $\mproj {\{ h, t\}} \cdot \mathcal T\cdot\mshift \ell t\cdot \mzero h \cdot \vec p \geq 0$ and $\mathfrak D \cdot \mshift \ell t\cdot \mzero h \cdot \vec p \geq 0$.
    The latter tells us that $\mshift \ell t\cdot \mzero h \cdot \vec p$
    is a sufficient vector we may use with our inductive hypothesis.

Applying the inductive hypothesis allows us to find 
an $M \in {\mathcal T}$ where the following holds:

\begin{tiny}
\begin{align*}
& \pot {(V, \ell \mapsto  \econs {v_h} {v_t}, h \mapsto v_h, t \mapsto v_t)} {(\Gamma,  \ell : \tlist \tau, h:\tau, t : \tlist \tau )} {\mshift \ell t \cdot \mzero h \cdot \vec p}
\\
&\geq \pot {(V, \ell \mapsto  \econs {v_h} {v_t}, h \mapsto v_h, t \mapsto v_t,\iret \mapsto v)} {(\Gamma,  \ell : \tlist \tau, h:\tau, t : \tlist \tau,\iret : \sigma)} {M \cdot\mshift \ell t \cdot \mzero h \cdot \vec p }   
& IH
\\
\implies & 
\pot {(V, \ell \mapsto  \econs {v_h} {v_t}, h \mapsto v_h, t \mapsto v_t)} {(\Gamma,  \ell : \tlist \tau, h:\tau, t : \tlist \tau )} {\mshift \ell t \cdot \mzero h \cdot \vec p}
\\
&\geq \pot {(V, \ell \mapsto  \econs {v_h} {v_t},\iret \mapsto v)} {(\Gamma,  \ell : \tlist \tau,\iret : \sigma)} {\mproj {\neg \{h,t\}} \cdot M \cdot\mshift \ell t \cdot  \mzero h \cdot \vec p }  
&  \mproj { \{ h, t\}} \cdot M \cdot \mshift \ell t   \mzero h \cdot \vec p \geq 0
\\
\implies & 
\pot {(V, h \mapsto v_h, t \mapsto v_t)} {(\Gamma, h:\tau, t : \tlist \tau )} {\mshift \ell t \cdot  \mzero h \cdot \vec p}
\\
&\geq \pot {(V,\iret \mapsto v)} {(\Gamma, \iret : \sigma)} {\mproj {\neg \{h,t\}} \cdot M \cdot\mshift \ell t \cdot  \mzero h \cdot\vec p }  
&  \textrm{0 potential}
\\
\implies & 
\pot {(V, \ell \mapsto \econs {v_h} {v_t}, h\mapsto v_h)} {(\Gamma,\ell : \tlist \tau, h:\tau )} { \mzero h \cdot \vec p}
\\
&\geq \pot {(V,\iret \mapsto v)} {(\Gamma, \iret : \sigma)} {\mproj {\neg \{h,t\}} \cdot M \cdot\mshift \ell t \cdot  \mzero h \cdot \vec p }  
& \mshift \ell t \; \textrm{conservative}
\\
\implies & 
\pot {(V, \ell \mapsto \econs {v_h} {v_t})} {(\Gamma,\ell : \tlist \tau )} {  \vec p}
\geq \pot {(V,\iret \mapsto v)} {(\Gamma, \iret : \sigma)} {\mproj {\neg \{h,t\}} \cdot M \cdot\mshift \ell t \cdot  \mzero h \cdot \vec p }  
& \textrm{0 potential}
\end{align*}
\end{tiny}

Thus $\mproj {\neg \{h,t\}} \cdot M \cdot\mshift \ell t \cdot  \mzero h
\in  \mproj {\neg \{h,t\}} \cdot \mathcal T \cdot\mshift \ell t \cdot  \mzero h$ is a sufficient matrix to witness this case.

\subsubsubsection{\textbf{let}}

Suppose \textit{E-Let} was the last evaluation rule applied. 
Then we know the rule's premises must hold
by inversion.

\[
\inferrule[E-Let]{
    \ejudgenp V {e_1} {v'}
    \\
    \ejudgenp {V,x\mapsto v'} {e_2} v
}{
    \ejudgenp V {\elet x {e_1} {e_2}} v
}
\]

Only one typing rule, \textit{M-Let}, types the let 
expression, so \textit{M-Let} must have been the last typing rule 
applied. Thus by inversion, both of its premises must hold as well.

\[
\inferrule[M-Let]{
    \tjudgecf {\Gamma} {e_1} \sigma {\mathcal S} {\mathfrak C}
    \\
    \tjudgecf {\Gamma, x:\sigma} {e_2} \tau {\mathcal T} {\mathfrak D}
}{
    \tjudgecf {\Gamma} {\elet x {e_1} {e_2}} \tau 
    {\mproj {\neg x} \cdot \mathcal T \cdot \mmove \iret x \cdot \mathcal S}
    {\mathfrak C \cup (\mproj x \cdot \mathcal T \cup \mathfrak D) \cdot \mmove \iret x \cdot \mathcal S}
}
\]

Because $(\mathfrak C \cup (\mproj x \cdot \mathcal T \cup \mathfrak D) \cdot \mmove \iret x \cdot \mathcal S) \cdot \vec p
\geq 0$ by assumption, we know in particular that $\mathfrak C \cdot \vec p \geq 0$. This allows us to use the inductive hypothesis with
the left side premises to find
a map $M_1 \in {\mathcal S}$ such that the following holds:

\begin{align*}
\pot {V} {\Gamma} {\vec p} \geq& \pot {(V,\iret \mapsto v')} {\Gamma,\iret:\sigma } {M_1\cdot \vec p} & IH
\\
= & \pot {(V,x \mapsto v')} {\Gamma,\iret:\sigma } {\mmove \iret x \cdot M_1 \cdot \vec p} 
& \textrm{relabelling}
\end{align*}

Now because $(\mathfrak C \cup (\mproj x \cdot \mathcal T \cup \mathfrak D) \cdot \mmove \iret x \cdot \mathcal S) \cdot \vec p
\geq 0$ by assumption, we also know that $\mproj x \cdot \mathcal T  \cdot \mmove \iret x \cdot \mathcal S \cdot \vec p \geq 0$ and 
$\mathfrak D \cdot \mmove \iret x \cdot \mathcal S \cdot \vec p
\geq 0$. Since $M_1 \in \mathcal S$, it follows that
$\mproj x \cdot \mathcal T  \cdot \mmove \iret x \cdot M_1 \cdot \vec p \geq 0$ and 
$\mathfrak D \cdot \mmove \iret x \cdot M_1 \cdot \vec p
\geq 0$. The latter tells us that $\mmove \iret x \cdot M_1 \cdot \vec p$ 
is a sufficient vector for applying the inductive hypothesis 
with our right side premises.

Applying the inductive hypothesis with the right side premises
 yields some $M_2 \in {\mathcal T }$ that we can
use to continue our inequalities as follows:

\begin{tiny}
\begin{align*}
\pot {(V,x \mapsto v')} {\Gamma,\iret:\sigma } {\mmove \iret x \cdot M_1 \cdot \vec p} 
\geq & \pot {(V,x \mapsto v',\iret \mapsto v)} {\Gamma,x:\sigma,\iret:\tau } {M_2 \cdot \mmove \iret x \cdot M_1 \cdot \vec p} & IH
\\
\geq & \pot {(V,\iret \mapsto v)} {\Gamma,\iret:\tau } {\mproj {\neg x} \cdot M_2 \cdot \mmove \iret x \cdot M_1 \cdot \vec p} & \mproj {x} \cdot M_2 \cdot  \mmove \iret x \cdot M_1 \cdot \vec p  \geq 0
\end{align*}
\end{tiny}

Thus $\mproj {\neg x} \cdot M_2 \cdot \mmove \iret x \cdot M_1 \in 
\mproj {\neg x} \cdot \mathcal T \cdot \mmove \iret x \cdot \mathcal S $ is a sufficient matrix to witness this case.

\subsubsubsection{\textbf{remaining cases}}

Each of the remaining cases is a base case not needing the
inductive hypothesis, where the desired inequality is met
with equality by construction. We show an example for 
one such case as follows:

Suppose that \textit{E-Fun} was the last evaluation rule applied.

\[
\inferrule[E-Fun]{
}{
    \ejudgenp V {\efun f x e} {\vclosure V f x e}
}\]

Only one typing rule, \textit{M-Fun}, types the function
expression, so \textit{M-Fun} must have been the last typing
rule applied.

\[
\inferrule[M-Fun]{
    \tjudgecf {\Gamma, f: \tfun \tau \sigma M, x: \tau} e \sigma {\mathcal S} {\mathfrak C}
    \\
     \mproj  {\{x, \iret,\iconst \} } \cdot \mathcal S \cdot \mzero {\dom{\Gamma}} 
     \geq \mproj {\{ x, \iret,\iconst \} } \cdot M \cdot \mmove x \iarg 
    \\
    \mproj  {\neg \{x, \iret,\iconst \} } \cdot \mathcal S \cdot \mzero {\dom{\Gamma}} \geq 0
    \\
    \mathfrak C \cdot \mzero {\dom{\Gamma}} \geq 0
}{
    \tjudgecf \Gamma {\efun f x e} {\tfun \tau \sigma M}
    {\{ \mi \}} {\emptyset }
}
\]

Because functions carry no potential and the annotations are unchanged 
by the identity function, the
desired conclusion is trivially met with an equality by picking 
the identity matrix $\mi$ as the witness.

\end{proof}

\section{Experimental Data}

\Cref{table:data1,table:data2,table:data3} contain the data recorded from our experiments,
as outlined in \Cref{sec:exp}.

\begin{table}[!ht]
    \centering
    \begin{tiny}
    \begin{tabular}{|l|l|l||l|l|l|l||l|l|l|l|}
    \hline
        \multicolumn{3}{c}{params} |& \multicolumn{4}{c}{new} |& \multicolumn{4}{c}{old} 
        \\ \hline
        d & c & $\ell$ & constr secs & solve secs & total secs & constrs & constr secs & solve sec & total secs & constrs \\ \hline
        1 & 0 & 0 & 0.000042 & 0.016539 & 0.016581 & 4 & 0.000026 & 0.015233 & 0.015259 & 12  \\ \hline
        1 & 0 & 1 & 0.000099 & 0.029961 & 0.030060 & 14 & 0.000041 & 0.015706 & 0.015747 & 53 \\ \hline
        1 & 0 & 2 & 0.000170 & 0.045621 & 0.045791 & 24 & 0.000074 & 0.015939 & 0.016013 & 94  \\ \hline
        1 & 0 & 3 & 0.000284 & 0.062229 & 0.062513 & 34 & 0.000094 & 0.016423 & 0.016517 & 135 \\ \hline
        1 & 0 & 4 & 0.000335 & 0.086318 & 0.086653 & 44 & 0.000464 & 0.021997 & 0.022461 & 176\\ \hline
        1 & 0 & 5 & 0.000380 & 0.089821 & 0.090201 & 54 & 0.000146 & 0.016866 & 0.017012 & 217 \\ \hline
        1 & 1 & 0 & 0.000025 & 0.014744 & 0.014769 & 4 & 0.000014 & 0.015327 & 0.015341 & 12  \\ \hline
        1 & 1 & 1 & 0.000107 & 0.032644 & 0.032751 & 17 & 0.000081 & 0.018657 & 0.018738 & 72  \\ \hline
        1 & 1 & 2 & 0.000209 & 0.045724 & 0.045933 & 30 & 0.000131 & 0.016680 & 0.016811 & 186  \\ \hline
        1 & 1 & 3 & 0.000268 & 0.057673 & 0.057941 & 43 & 0.000238 & 0.017584 & 0.017822 & 352  \\ \hline
        1 & 1 & 4 & 0.000389 & 0.087940 & 0.088329 & 56 & 0.000644 & 0.024758 & 0.025402 & 570  \\ \hline
        1 & 1 & 5 & 0.001474 & 0.114994 & 0.116468 & 69 & 0.000856 & 0.024653 & 0.025509 & 840 \\ \hline
        1 & 2 & 0 & 0.000019 & 0.015084 & 0.015103 & 4 & 0.000017 & 0.015241 & 0.015258 & 12 \\ \hline
        1 & 2 & 1 & 0.000119 & 0.030124 & 0.030243 & 21 & 0.000088 & 0.016134 & 0.016222 & 95  \\ \hline
        1 & 2 & 2 & 0.000212 & 0.042785 & 0.042997 & 38 & 0.000219 & 0.016977 & 0.017196 & 321 \\ \hline
        1 & 2 & 3 & 0.000308 & 0.056932 & 0.057240 & 55 & 0.000611 & 0.019885 & 0.020496 & 831  \\ \hline
        1 & 2 & 4 & 0.000543 & 0.083748 & 0.084291 & 72 & 0.001632 & 0.030908 & 0.032540 & 1912 \\ \hline
        1 & 2 & 5 & 0.000629 & 0.104308 & 0.104937 & 89 & 0.005318 & 0.057943 & 0.063261 & 4132 \\ \hline
        1 & 3 & 0 & 0.000018 & 0.016538 & 0.016556 & 4 & 0.000017 & 0.016756 & 0.016773 & 12 \\ \hline
        1 & 3 & 1 & 0.000149 & 0.031353 & 0.031502 & 21 & 0.000084 & 0.015582 & 0.015666 & 115  \\ \hline
        1 & 3 & 2 & 0.000415 & 0.077578 & 0.077993 & 38 & 0.000374 & 0.022728 & 0.023102 & 491  \\ \hline
        1 & 3 & 3 & 0.000365 & 0.060048 & 0.060413 & 55 & 0.001375 & 0.029278 & 0.030653 & 1691 \\ \hline
        1 & 3 & 4 & 0.000849 & 0.144139 & 0.144988 & 72 & 0.010170 & 0.184897 & 0.195067 & 5351  \\ \hline
        1 & 3 & 5 & 0.000605 & 0.091483 & 0.092088 & 89 & 0.016646 & 0.147977 & 0.164623 & 16400 \\ \hline
        1 & 4 & 0 & 0.000017 & 0.015032 & 0.015049 & 4 & 0.000015 & 0.016025 & 0.016040 & 12 \\ \hline
        1 & 4 & 1 & 0.000216 & 0.029888 & 0.030104 & 21 & 0.000132 & 0.017826 & 0.017958 & 135 \\ \hline
        1 & 4 & 2 & 0.000361 & 0.047575 & 0.047936 & 38 & 0.000521 & 0.022471 & 0.022992 & 702  \\ \hline
        1 & 4 & 3 & 0.000438 & 0.071367 & 0.071805 & 55 & 0.004783 & 0.084018 & 0.088801 & 3048 \\ \hline
        1 & 4 & 4 & 0.000559 & 0.076119 & 0.076678 & 72 & 0.012303 & 0.111350 & 0.123653 & 12510 \\ \hline
        1 & 4 & 5 & 0.000818 & 0.087661 & 0.088479 & 89 & 0.056572 & 0.602399 & 0.658971 & 50434  \\ \hline
        1 & 5 & 0 & 0.000019 & 0.015040 & 0.015059 & 4 & 0.000015 & 0.015895 & 0.015910 & 12  \\ \hline
        1 & 5 & 1 & 0.000225 & 0.036546 & 0.036771 & 21 & 0.000101 & 0.018815 & 0.018916 & 155 \\ \hline
        1 & 5 & 2 & 0.000290 & 0.043415 & 0.043705 & 38 & 0.000833 & 0.022671 & 0.023504 & 953  \\ \hline
        1 & 5 & 3 & 0.000583 & 0.061408 & 0.061991 & 55 & 0.006523 & 0.051278 & 0.057801 & 5031  \\ \hline
        1 & 5 & 4 & 0.003645 & 0.085588 & 0.089233 & 72 & 0.025974 & 0.278086 & 0.304060 & 25502\\ \hline
        1 & 5 & 5 & 0.001311 & 0.086741 & 0.088052 & 89 & 0.142253 & 1.285613 & 1.427866 & 127944  \\ \hline
        2 & 0 & 0 & 0.000028 & 0.015328 & 0.015356 & 9 & 0.000029 & 0.015573 & 0.015602 & 20  \\ \hline
        2 & 0 & 1 & 0.000289 & 0.033261 & 0.033550 & 29 & 0.000131 & 0.019285 & 0.019416 & 129 \\ \hline
        2 & 0 & 2 & 0.000502 & 0.047067 & 0.047569 & 49 & 0.000160 & 0.017428 & 0.017588 & 238  \\ \hline
        2 & 0 & 3 & 0.000729 & 0.059772 & 0.060501 & 69 & 0.000223 & 0.017761 & 0.017984 & 346  \\ \hline
        2 & 0 & 4 & 0.001074 & 0.091373 & 0.092447 & 89 & 0.000325 & 0.022354 & 0.022679 & 454\\ \hline
        2 & 0 & 5 & 0.001716 & 0.114943 & 0.116659 & 109 & 0.000558 & 0.024100 & 0.024658 & 562\\ \hline
        2 & 1 & 0 & 0.000028 & 0.015250 & 0.015278 & 9 & 0.000017 & 0.015819 & 0.015836 & 20 \\ \hline
        2 & 1 & 1 & 0.000360 & 0.031529 & 0.031889 & 35 & 0.000136 & 0.017132 & 0.017268 & 193  \\ \hline
        2 & 1 & 2 & 0.000686 & 0.046000 & 0.046686 & 61 & 0.000500 & 0.021283 & 0.021783 & 623  \\ \hline
        2 & 1 & 3 & 0.001164 & 0.109110 & 0.110274 & 87 & 0.001144 & 0.027268 & 0.028412 & 1410  \\ \hline
        2 & 1 & 4 & 0.001434 & 0.076398 & 0.077832 & 113 & 0.003020 & 0.044966 & 0.047986 & 2660 \\ \hline
        2 & 1 & 5 & 0.001770 & 0.112141 & 0.113911 & 139 & 0.005409 & 0.059942 & 0.065351 & 4476  \\ \hline
        2 & 2 & 0 & 0.000025 & 0.015468 & 0.015493 & 9 & 0.000018 & 0.015634 & 0.015652 & 20  \\ \hline
        2 & 2 & 1 & 0.000429 & 0.031337 & 0.031766 & 44 & 0.000165 & 0.017853 & 0.018018 & 254  \\ \hline
        2 & 2 & 2 & 0.000807 & 0.046065 & 0.046872 & 79 & 0.000981 & 0.024991 & 0.025972 & 1204  \\ \hline
        2 & 2 & 3 & 0.001801 & 0.062295 & 0.064096 & 114 & 0.003758 & 0.052979 & 0.056737 & 4150  \\ \hline
        2 & 2 & 4 & 0.002516 & 0.080507 & 0.083023 & 149 & 0.013180 & 0.177627 & 0.190807 & 12238 \\ \hline
        2 & 2 & 5 & 0.002774 & 0.095361 & 0.098134 & 184 & 0.037644 & 0.463047 & 0.500691 & 32885  \\ \hline
        2 & 3 & 0 & 0.000028 & 0.017961 & 0.017989 & 9 & 0.000018 & 0.017703 & 0.017721 & 20 \\ \hline
        2 & 3 & 1 & 0.000589 & 0.034955 & 0.035544 & 44 & 0.000203 & 0.020655 & 0.020858 & 316  \\ \hline
        2 & 3 & 2 & 0.001131 & 0.052838 & 0.053969 & 79 & 0.001976 & 0.035812 & 0.037788 & 1989  \\ \hline
        2 & 3 & 3 & 0.001467 & 0.060390 & 0.061857 & 114 & 0.009366 & 0.106147 & 0.115513 & 9425  \\ \hline
        2 & 3 & 4 & 0.004764 & 0.078818 & 0.083582 & 149 & 0.047127 & 0.589632 & 0.636759 & 39081 \\ \hline
        2 & 3 & 5 & 0.002945 & 0.094591 & 0.097536 & 184 & 0.175089 & 1.930193 & 2.105282 & 150189  \\ \hline
        2 & 4 & 0 & 0.000023 & 0.015724 & 0.015747 & 9 & 0.000015 & 0.015552 & 0.015567 & 20  \\ \hline
        2 & 4 & 1 & 0.000628 & 0.030262 & 0.030890 & 44 & 0.000276 & 0.018271 & 0.018547 & 377 \\ \hline
        2 & 4 & 2 & 0.001415 & 0.053228 & 0.054643 & 79 & 0.002895 & 0.048787 & 0.051682 & 2970 \\ \hline
        2 & 4 & 3 & 0.001849 & 0.066815 & 0.068664 & 114 & 0.020787 & 0.248645 & 0.269432 & 18060 \\ \hline
        2 & 4 & 4 & 0.003051 & 0.079288 & 0.082339 & 149 & 0.121902 & 2.613878 & 2.735780 & 97380  \\ \hline
        2 & 4 & 5 & 0.003527 & 0.095953 & 0.099480 & 184 & 0.603298 & 10.648155 & 11.251453 & 490539 \\ \hline
        2 & 5 & 0 & 0.000026 & 0.016016 & 0.016042 & 9 & 0.000020 & 0.016125 & 0.016145 & 20 \\ \hline
        2 & 5 & 1 & 0.000678 & 0.033611 & 0.034289 & 44 & 0.000302 & 0.019871 & 0.020173 & 438 \\ \hline
        2 & 5 & 2 & 0.001414 & 0.045110 & 0.046524 & 79 & 0.003294 & 0.052815 & 0.056109 & 4154 \\ \hline
        2 & 5 & 3 & 0.003176 & 0.093567 & 0.096743 & 114 & 0.055972 & 0.501996 & 0.557968 & 30940 \\ \hline
        2 & 5 & 4 & 0.003303 & 0.081001 & 0.084304 & 149 & 0.240742 & 3.001745 & 3.242487 & 205845 \\ \hline
        2 & 5 & 5 & 0.006122 & 0.093704 & 0.099826 & 184 & 1.503743 & 17.240749 & 18.744492 & 1285286 \\ \hline
    \end{tabular}
    \end{tiny}
    \caption{experimental data for $d=1,2$}
    \label{table:data1}
\end{table}

\begin{table}[!ht]
    \centering
    \begin{tiny}
    \begin{tabular}{|l|l|l||l|l|l|l||l|l|l|l|}
    \hline
        \multicolumn{3}{c}{params} |& \multicolumn{4}{c}{new} |& \multicolumn{4}{c}{old} 
        \\ \hline
        d & c & $\ell$ & constr secs & solve secs & total secs & constrs & constr secs & solve sec & total secs & constrs \\ \hline
        3 & 0 & 0 & 0.000041 & 0.015634 & 0.015675 & 16 & 0.000026 & 0.015668 & 0.015694 & 28 \\ \hline
        3 & 0 & 1 & 0.001046 & 0.032759 & 0.033805 & 43 & 0.000186 & 0.018984 & 0.019170 & 236  \\ \hline
        3 & 0 & 2 & 0.001450 & 0.049265 & 0.050715 & 70 & 0.000379 & 0.020762 & 0.021141 & 443  \\ \hline
        3 & 0 & 3 & 0.001887 & 0.062912 & 0.064799 & 97 & 0.000490 & 0.021987 & 0.022477 & 649  \\ \hline
        3 & 0 & 4 & 0.002996 & 0.092591 & 0.095587 & 124 & 0.000699 & 0.026594 & 0.027293 & 856  \\ \hline
        3 & 0 & 5 & 0.003764 & 0.113372 & 0.117136 & 151 & 0.001065 & 0.025972 & 0.027037 & 1063  \\ \hline
        3 & 1 & 0 & 0.000036 & 0.015988 & 0.016024 & 16 & 0.000020 & 0.016161 & 0.016181 & 28  \\ \hline
        3 & 1 & 1 & 0.001061 & 0.032295 & 0.033356 & 59 & 0.000331 & 0.018910 & 0.019241 & 357  \\ \hline
        3 & 1 & 2 & 0.002349 & 0.049491 & 0.051840 & 102 & 0.001315 & 0.031187 & 0.032502 & 1395 \\ \hline
        3 & 1 & 3 & 0.003493 & 0.063860 & 0.067353 & 145 & 0.003248 & 0.059259 & 0.062507 & 3758 \\ \hline
        3 & 1 & 4 & 0.008033 & 0.098151 & 0.106184 & 188 & 0.009770 & 0.126782 & 0.136552 & 8266 \\ \hline
        3 & 1 & 5 & 0.006617 & 0.100908 & 0.107525 & 231 & 0.020647 & 0.202641 & 0.223288 & 15948  \\ \hline
        3 & 2 & 0 & 0.000034 & 0.016045 & 0.016079 & 16 & 0.000017 & 0.015359 & 0.015376 & 28  \\ \hline
        3 & 2 & 1 & 0.001706 & 0.033962 & 0.035668 & 75 & 0.000358 & 0.020168 & 0.020526 & 479  \\ \hline
        3 & 2 & 2 & 0.003464 & 0.048022 & 0.051486 & 134 & 0.003168 & 0.046793 & 0.049961 & 2913  \\ \hline
        3 & 2 & 3 & 0.005143 & 0.069853 & 0.074996 & 193 & 0.014636 & 0.182602 & 0.197238 & 12759  \\ \hline
        3 & 2 & 4 & 0.008598 & 0.083686 & 0.092284 & 252 & 0.053517 & 1.009344 & 1.062861 & 46559 \\ \hline
        3 & 2 & 5 & 0.010271 & 0.100567 & 0.110838 & 311 & 0.199687 & 3.691999 & 3.891686 & 151105 \\ \hline
        3 & 3 & 0 & 0.000045 & 0.023406 & 0.023451 & 16 & 0.000021 & 0.021076 & 0.021097 & 28 \\ \hline
        3 & 3 & 1 & 0.001888 & 0.036658 & 0.038546 & 75 & 0.000878 & 0.022383 & 0.023261 & 601  \\ \hline
        3 & 3 & 2 & 0.004736 & 0.050563 & 0.055299 & 134 & 0.004270 & 0.074808 & 0.079078 & 5013 \\ \hline
        3 & 3 & 3 & 0.006328 & 0.068159 & 0.074487 & 193 & 0.037331 & 0.611558 & 0.648889 & 30823  \\ \hline
        3 & 3 & 4 & 0.018882 & 0.084396 & 0.103278 & 252 & 0.187965 & 5.052531 & 5.240496 & 160364 \\ \hline
        3 & 3 & 5 & 0.022301 & 0.101693 & 0.123994 & 311 & 0.931370 & 18.897349 & 19.828719 & 749203  \\ \hline
        3 & 4 & 0 & 0.000039 & 0.019391 & 0.019430 & 16 & 0.000026 & 0.017685 & 0.017711 & 28 \\ \hline
        3 & 4 & 1 & 0.002620 & 0.038494 & 0.041114 & 75 & 0.000599 & 0.025804 & 0.026403 & 721 \\ \hline
        3 & 4 & 2 & 0.008145 & 0.050352 & 0.058497 & 134 & 0.007200 & 0.110314 & 0.117514 & 7654 \\ \hline
        3 & 4 & 3 & 0.007776 & 0.068316 & 0.076092 & 193 & 0.055860 & 1.653075 & 1.708935 & 60979 \\ \hline
        3 & 4 & 4 & 0.008763 & 0.083909 & 0.092672 & 252 & 0.440956 & 14.162719 & 14.603675 & 414116 \\ \hline
        3 & 4 & 5 & 0.014517 & 0.105401 & 0.119918 & 311 & 3.195450 & 171.840712 & 175.036162 & 2537242 \\ \hline
        3 & 5 & 0 & 0.000036 & 0.015974 & 0.016010 & 16 & 0.000017 & 0.015511 & 0.015528 & 28 \\ \hline
        3 & 5 & 1 & 0.002837 & 0.033307 & 0.036144 & 75 & 0.000654 & 0.022925 & 0.023579 & 842 \\ \hline
        3 & 5 & 2 & 0.006123 & 0.050831 & 0.056954 & 134 & 0.019261 & 0.154263 & 0.173524 & 10872 \\ \hline
        3 & 5 & 3 & 0.011121 & 0.067044 & 0.078165 & 193 & 0.121545 & 2.997829 & 3.119374 & 106554  \\ \hline
        3 & 5 & 4 & 0.011861 & 0.084976 & 0.096837 & 252 & 1.122581 & 53.309085 & 54.431666 & 893955 \\ \hline
        3 & 5 & 5 & 0.014272 & 0.097983 & 0.112255 & 311 & 9.312985 & 1179.742946 & 1189.055931 & 6785087  \\ \hline
        4 & 0 & 0 & 0.000069 & 0.019897 & 0.019966 & 25 & 0.000044 & 0.016803 & 0.016847 & 36  \\ \hline
        4 & 0 & 1 & 0.001729 & 0.033189 & 0.034918 & 59 & 0.000337 & 0.019493 & 0.019830 & 365  \\ \hline
        4 & 0 & 2 & 0.003536 & 0.056949 & 0.060485 & 93 & 0.000665 & 0.022510 & 0.023175 & 694 \\ \hline
        4 & 0 & 3 & 0.005667 & 0.076041 & 0.081708 & 127 & 0.000871 & 0.027398 & 0.028269 & 1024  \\ \hline
        4 & 0 & 4 & 0.015344 & 0.128815 & 0.144159 & 161 & 0.001213 & 0.032686 & 0.033899 & 1353 \\ \hline
        4 & 0 & 5 & 0.010312 & 0.107239 & 0.117551 & 195 & 0.001545 & 0.032817 & 0.034362 & 1685  \\ \hline
        4 & 1 & 0 & 0.000059 & 0.018750 & 0.018809 & 25 & 0.000025 & 0.018959 & 0.018984 & 36  \\ \hline
        4 & 1 & 1 & 0.005002 & 0.047072 & 0.052074 & 84 & 0.000478 & 0.059413 & 0.059891 & 568 \\ \hline
        4 & 1 & 2 & 0.009083 & 0.055694 & 0.064777 & 143 & 0.002287 & 0.049415 & 0.051702 & 2606  \\ \hline
        4 & 1 & 3 & 0.013858 & 0.068480 & 0.082338 & 202 & 0.009116 & 0.121333 & 0.130449 & 8179 \\ \hline
        4 & 1 & 4 & 0.019216 & 0.092336 & 0.111552 & 261 & 0.025725 & 0.343922 & 0.369647 & 20762  \\ \hline
        4 & 1 & 5 & 0.029570 & 0.150571 & 0.180141 & 320 & 0.090916 & 1.158684 & 1.249600 & 45671  \\ \hline
        4 & 2 & 0 & 0.000059 & 0.026470 & 0.026529 & 25 & 0.000033 & 0.024135 & 0.024168 & 36 \\ \hline
        4 & 2 & 1 & 0.006506 & 0.036330 & 0.042836 & 109 & 0.000667 & 0.024205 & 0.024872 & 771  \\ \hline
        4 & 2 & 2 & 0.013719 & 0.053968 & 0.067687 & 193 & 0.006065 & 0.090156 & 0.096221 & 5731  \\ \hline
        4 & 2 & 3 & 0.027344 & 0.080500 & 0.107844 & 277 & 0.029282 & 0.640377 & 0.669659 & 30523 \\ \hline
        4 & 2 & 4 & 0.028854 & 0.091598 & 0.120452 & 361 & 0.134407 & 4.517918 & 4.652325 & 133754  \\ \hline
        4 & 2 & 5 & 0.037787 & 0.112178 & 0.149965 & 445 & 0.543178 & 24.347409 & 24.890587 & 512519  \\ \hline
        4 & 3 & 0 & 0.000053 & 0.017022 & 0.017075 & 25 & 0.000025 & 0.016221 & 0.016246 & 36  \\ \hline
        4 & 3 & 1 & 0.009729 & 0.045220 & 0.054949 & 109 & 0.001295 & 0.030739 & 0.032034 & 972  \\ \hline
        4 & 3 & 2 & 0.018507 & 0.054518 & 0.073025 & 193 & 0.009415 & 0.160940 & 0.170355 & 10074  \\ \hline
        4 & 3 & 3 & 0.034511 & 0.075975 & 0.110486 & 277 & 0.093669 & 2.533223 & 2.626892 & 76585  \\ \hline
        4 & 3 & 4 & 0.042081 & 0.087974 & 0.130055 & 361 & 0.612588 & 21.951089 & 22.563677 & 483261 \\ \hline
        4 & 3 & 5 & 0.051126 & 0.114073 & 0.165199 & 445 & 4.248675 & 448.071415 & 452.320090 & 2680916 \\ \hline
        4 & 4 & 0 & 0.000049 & 0.017285 & 0.017334 & 25 & 0.000021 & 0.016384 & 0.016405 & 36 \\ \hline
        4 & 4 & 1 & 0.011328 & 0.035795 & 0.047123 & 109 & 0.001075 & 0.028545 & 0.029620 & 1173  \\ \hline
        4 & 4 & 2 & 0.034361 & 0.053880 & 0.088241 & 193 & 0.026270 & 0.292175 & 0.318445 & 15596 \\ \hline
        4 & 4 & 3 & 0.038971 & 0.070234 & 0.109205 & 277 & 0.146051 & 7.262634 & 7.408685 & 154511  \\ \hline
        4 & 4 & 4 & 0.044262 & 0.091807 & 0.136069 & 361 & 1.697200 & 167.918893 & 169.616093 & 1276812 \\ \hline
        4 & 4 & 5 & 0.056506 & 0.141851 & 0.198357 & 445 & 12.724109 & 4700.986052 & 4713.710161 & 9301888 \\ \hline
        4 & 5 & 0 & 0.000057 & 0.017507 & 0.017564 & 25 & 0.000025 & 0.016347 & 0.016372 & 36  \\ \hline
        4 & 5 & 1 & 0.012953 & 0.034416 & 0.047369 & 109 & 0.001228 & 0.030184 & 0.031412 & 1373  \\ \hline
        4 & 5 & 2 & 0.028344 & 0.053067 & 0.081411 & 193 & 0.019104 & 0.401896 & 0.421000 & 22350 \\ \hline
        4 & 5 & 3 & 0.044963 & 0.080593 & 0.125556 & 277 & 0.302339 & 17.910969 & 18.213308 & 273266 \\ \hline
        4 & 5 & 4 & 0.065743 & 0.087002 & 0.152745 & 361 & 3.626365 & 588.597753 & 592.224118 & 2793847  \\ \hline
        4 & 5 & 5 & 0.069873 & 0.114544 & 0.184417 & 445 & 36.485721 & 70570.439041 & 70606.924762 & 25225105 \\ \hline
    \end{tabular}
    \end{tiny}
    \caption{experimental data for $d=3,4$}
    \label{table:data2}
\end{table}

\begin{table}[!ht]
    \centering
    \begin{tiny}
    \begin{tabular}{|l|l|l||l|l|l|l||l|l|l|l|}
    \hline
        \multicolumn{3}{c}{params} |& \multicolumn{4}{c}{new} |& \multicolumn{4}{c}{old} 
        \\ \hline
        d & c & $\ell$ & constr secs & solve secs & total secs & constrs & constr secs & solve sec & total secs & constrs \\ \hline
        5 & 0 & 0 & 0.000077 & 0.017486 & 0.017563 & 36 & 0.000037 & 0.016726 & 0.016763 & 44  \\ \hline
        5 & 0 & 1 & 0.003436 & 0.039762 & 0.043198 & 83 & 0.000409 & 0.027827 & 0.028236 & 532  \\ \hline
        5 & 0 & 2 & 0.006980 & 0.061981 & 0.068961 & 130 & 0.000961 & 0.032137 & 0.033098 & 1020  \\ \hline
        5 & 0 & 3 & 0.024448 & 0.137288 & 0.161736 & 177 & 0.001332 & 0.039969 & 0.041301 & 1507 \\ \hline
        5 & 0 & 4 & 0.013967 & 0.097683 & 0.111650 & 224 & 0.001894 & 0.037573 & 0.039467 & 1995 \\ \hline
        5 & 0 & 5 & 0.015675 & 0.114982 & 0.130657 & 271 & 0.002366 & 0.041321 & 0.043687 & 2482 \\ \hline
        5 & 1 & 0 & 0.000066 & 0.017811 & 0.017877 & 36 & 0.000029 & 0.016461 & 0.016490 & 44  \\ \hline
        5 & 1 & 1 & 0.009268 & 0.036315 & 0.045583 & 119 & 0.000722 & 0.025874 & 0.026596 & 826  \\ \hline
        5 & 1 & 2 & 0.018078 & 0.056897 & 0.074975 & 202 & 0.004226 & 0.076642 & 0.080868 & 4345 \\ \hline
        5 & 1 & 3 & 0.029984 & 0.076846 & 0.106830 & 285 & 0.020346 & 0.285027 & 0.305373 & 15644  \\ \hline
        5 & 1 & 4 & 0.036824 & 0.102572 & 0.139396 & 368 & 0.058999 & 1.694275 & 1.753274 & 45265  \\ \hline
        5 & 1 & 5 & 0.054419 & 0.113217 & 0.167636 & 451 & 0.135150 & 5.011172 & 5.146322 & 112556  \\ \hline
        5 & 2 & 0 & 0.000076 & 0.017852 & 0.017928 & 36 & 0.000027 & 0.016072 & 0.016099 & 44  \\ \hline
        5 & 2 & 1 & 0.016984 & 0.044732 & 0.061716 & 155 & 0.000843 & 0.032951 & 0.033794 & 1128 \\ \hline
        5 & 2 & 2 & 0.029933 & 0.058477 & 0.088410 & 274 & 0.010980 & 0.171963 & 0.182943 & 9913 \\ \hline
        5 & 2 & 3 & 0.040894 & 0.083943 & 0.124837 & 393 & 0.078169 & 2.127854 & 2.206023 & 62382 \\ \hline
        5 & 2 & 4 & 0.053890 & 0.095180 & 0.149070 & 512 & 0.379000 & 16.143426 & 16.522426 & 320332 \\ \hline
        5 & 2 & 5 & 0.080353 & 0.129780 & 0.210133 & 631 & 2.058113 & 213.112385 & 215.170498 & 1421648 \\ \hline
        5 & 3 & 0 & 0.000063 & 0.018092 & 0.018155 & 36 & 0.000025 & 0.018186 & 0.018211 & 44  \\ \hline
        5 & 3 & 1 & 0.019366 & 0.038255 & 0.057621 & 155 & 0.001215 & 0.033221 & 0.034436 & 1432  \\ \hline
        5 & 3 & 2 & 0.038226 & 0.058702 & 0.096928 & 274 & 0.017728 & 0.349406 & 0.367134 & 17703  \\ \hline
        5 & 3 & 3 & 0.055996 & 0.082414 & 0.138410 & 393 & 0.189165 & 6.806947 & 6.996112 & 160536  \\ \hline
        5 & 3 & 4 & 0.081476 & 0.104385 & 0.185861 & 512 & 1.654466 & 225.653820 & 227.308286 & 1195445 \\ \hline
        5 & 3 & 5 & 0.109081 & 0.249445 & 0.358526 & 631 & 10.321564 & 6194.622692 & 6204.944256 & 7718071  \\ \hline
        5 & 4 & 0 & 0.000143 & 0.017999 & 0.018142 & 36 & 0.000036 & 0.016671 & 0.016707 & 44 \\ \hline
        5 & 4 & 1 & 0.027136 & 0.039636 & 0.066772 & 155 & 0.001423 & 0.034885 & 0.036308 & 1727 \\ \hline
        5 & 4 & 2 & 0.054744 & 0.057714 & 0.112458 & 274 & 0.034404 & 0.634923 & 0.669327 & 27643 \\ \hline
        5 & 4 & 3 & 0.073931 & 0.094613 & 0.168544 & 393 & 0.645488 & 36.678300 & 37.323788 & 328136  \\ \hline
        5 & 4 & 4 & 0.171072 & 0.199687 & 0.370759 & 512 & 5.518761 & 1182.287350 & 1187.806111 & 3208306  \\ \hline
        5 & 4 & 5 & 0.131733 & 0.186035 & 0.317768 & 631 & 35.674655 & 154367.904023 & 154403.578678 & 27243494 \\ \hline
        5 & 5 & 0 & 0.000143 & 0.018394 & 0.018537 & 36 & 0.000029 & 0.016771 & 0.016800 & 44 \\ \hline
        5 & 5 & 1 & 0.028720 & 0.038908 & 0.067628 & 155 & 0.001854 & 0.039140 & 0.040994 & 2029\\ \hline
        5 & 5 & 2 & 0.066871 & 0.056814 & 0.123685 & 274 & 0.042691 & 1.075947 & 1.118638 & 39857\\ \hline
        5 & 5 & 3 & 0.088594 & 0.132821 & 0.221415 & 393 & 0.581227 & 84.127568 & 84.708795 & 585037  \\ \hline
        5 & 5 & 4 & 0.435554 & 0.603455 & 1.039009 & 512 & 36.992054 & 5770.545060 & 5807.537114 & 7086341  \\ \hline
        5 & 5 &5 & 0.155609 & 0.146074 & 0.301683 & 631 & timeout & timeout & timeout & timeout \\
        \hline 
        6 & 0 & 0 & 0.000085 & 0.019476 & 0.019561 & 49 & 0.000030 & 0.016957 & 0.016987 & 52  \\ \hline
        6 & 0 & 1 & 0.006856 & 0.053433 & 0.060289 & 121 & 0.000972 & 0.032932 & 0.033904 & 718 \\ \hline
        6 & 0 & 2 & 0.010274 & 0.062646 & 0.072920 & 193 & 0.001603 & 0.034992 & 0.036595 & 1385  \\ \hline
        6 & 0 & 3 & 0.015783 & 0.083983 & 0.099766 & 265 & 0.001778 & 0.044023 & 0.045801 & 2052  \\ \hline
        6 & 0 & 4 & 0.022587 & 0.103950 & 0.126537 & 337 & 0.002568 & 0.049037 & 0.051605 & 2719  \\ \hline
        6 & 0 & 5 & 0.023988 & 0.124171 & 0.148159 & 409 & 0.003221 & 0.053733 & 0.056954 & 3386  \\ \hline
        6 & 1 & 0 & 0.000084 & 0.019479 & 0.019563 & 49 & 0.000030 & 0.018554 & 0.018584 & 52  \\ \hline
        6 & 1 & 1 & 0.010019 & 0.043066 & 0.053085 & 163 & 0.001360 & 0.032638 & 0.033998 & 1130  \\ \hline
        6 & 1 & 2 & 0.018473 & 0.059746 & 0.078219 & 277 & 0.006969 & 0.124090 & 0.131059 & 6710  \\ \hline
        6 & 1 & 3 & 0.031780 & 0.121696 & 0.153476 & 391 & 0.052798 & 0.826922 & 0.879720 & 27306  \\ \hline
        6 & 1 & 4 & 0.047615 & 0.106880 & 0.154495 & 505 & 0.150234 & 3.723004 & 3.873238 & 89024  \\ \hline
        6 & 1 & 5 & 0.042664 & 0.126581 & 0.169245 & 619 & 0.433976 & 9.306156 & 9.740132 & 247862  \\ \hline
        6 & 2 & 0 & 0.000089 & 0.019588 & 0.019677 & 49 & 0.000030 & 0.018670 & 0.018700 & 52  \\ \hline
        6 & 2 & 1 & 0.012102 & 0.043387 & 0.055489 & 212 & 0.002004 & 0.037488 & 0.039492 & 1550  \\ \hline
        6 & 2 & 2 & 0.027652 & 0.062923 & 0.090575 & 375 & 0.017866 & 0.324983 & 0.342849 & 15722  \\ \hline
        6 & 2 & 3 & 0.033779 & 0.084825 & 0.118604 & 538 & 0.133347 & 5.242922 & 5.376269 & 114325  \\ \hline
        6 & 2 & 4 & 0.056717 & 0.108987 & 0.165704 & 701 & 0.908094 & 67.925591 & 68.833685 & 674646  \\ \hline
        6 & 2 & 5 & 0.064283 & 0.204992 & 0.269275 & 864 & 5.001324 & 1296.569909 & 1301.571233 & 3412730 \\ \hline
        6 & 3 & 0 & 0.000668 & 0.019511 & 0.020179 & 49 & 0.000032 & 0.016984 & 0.017016 & 52  \\ \hline
        6 & 3 & 1 & 0.017178 & 0.042157 & 0.059335 & 212 & 0.001928 & 0.045486 & 0.047414 & 1969  \\ \hline
        6 & 3 & 2 & 0.028368 & 0.063418 & 0.091786 & 375 & 0.027076 & 0.725203 & 0.752279 & 28390  \\ \hline
        6 & 3 & 3 & 0.070527 & 0.114887 & 0.185414 & 538 & 0.438994 & 21.245849 & 21.684843 & 299542  \\ \hline
        6 & 3 & 4 & 0.058867 & 0.114669 & 0.173536 & 701 & 3.070279 & 906.614686 & 909.684965 & 2576633  \\ \hline
        6 & 3 & 5 & 0.306451 & 0.675611 & 0.982062 & 864 & 71.489372 & 79469.399906 & 79540.889278 & 19035507 \\ \hline
        6 & 4 & 0 & 0.000104 & 0.023031 & 0.023135 & 49 & 0.000032 & 0.018054 & 0.018086 & 52 \\ \hline
        6 & 4 & 1 & 0.018192 & 0.041655 & 0.059847 & 212 & 0.002206 & 0.046417 & 0.048623 & 2384  \\ \hline
        6 & 4 & 2 & 0.035369 & 0.063444 & 0.098813 & 375 & 0.040629 & 1.421071 & 1.461700 & 44632 \\ \hline
        6 & 4 & 3 & 0.054724 & 0.095273 & 0.149997 & 538 & 0.744372 & 73.169898 & 73.914270 & 618038  \\ \hline
        6 & 4 & 4 & 0.098837 & 0.155009 & 0.253846 & 701 & 8.922452 & 6845.493612 & 6854.416064 & 6994063 \\ \hline
        6& 4 & 5 & 0.103610 & 0.150520 & 0.254130 & 864 &  timeout & timeout & timeout & timeout \\ \hline
        6 & 5 & 0 & 0.000111 & 0.026519 & 0.026630 & 49 & 0.000045 & 0.018480 & 0.018525 & 52  \\ \hline
        6 & 5 & 1 & 0.023738 & 0.042042 & 0.065780 & 212 & 0.004236 & 0.050720 & 0.054956 & 2804  \\ \hline
        6 & 5 & 2 & 0.041496 & 0.065930 & 0.107426 & 375 & 0.063495 & 2.458531 & 2.522026 & 64600  \\ \hline
        6 & 5 & 3 & 0.063459 & 0.091847 & 0.155306 & 538 & 1.080118 & 212.582821 & 213.662939 & 1108168  \\ \hline
        6&5&4&0.083420&0.110104&0.193524&701
        & timeout & timeout & timeout & timeout \\ \hline
        6&5&5&0.110426&0.139084&0.249510&864
        & timeout & timeout & timeout & timeout \\\hline
    \end{tabular}
    \end{tiny}
    \caption{experimental data for $d=5,6$}
    \label{table:data3}
\end{table}

\end{document}